\documentclass[12pt]{article}

\newenvironment{wileykeywords}{\textsf{Keywords:}\hspace{\stretch{1}}}{\hspace{\stretch{1}}\rule{1ex}{1ex}}
\newcommand*\samethanks[1][\value{footnote}]{\footnotemark[#1]}

\usepackage{amsmath,amssymb,tikz,enumerate,epstopdf,amsthm}
\usepackage{graphicx}
\usepackage{color}
\usepackage{dcolumn}
\usepackage{bm}
\usepackage[numbers,super,comma,sort&compress]{natbib}

\newtheorem*{Proposition}{Proposition}
\definecolor{background-color}{gray}{0.98}

\title{\textit{In-situ} Data Analysis of Protein Folding Trajectories}
\author{Travis Johnston\thanks{Department of Computer and Information Sciences, University of Delaware}, 
             Boyu Zhang\samethanks[1],
			 Adam Liwo\thanks{Department of Chemistry, University of Gdansk, Poland},
	   Silvia Crivelli\thanks{Lawrence Berkeley National Laboratory},
	    Michela Taufer\samethanks[1]}

\begin{document}

\maketitle

\begin{abstract}
The transition from petascale to exascale computers is characterized
by substantial changes in the computer architectures and
technologies. The research community relying on computational
simulations is being forced to revisit the algorithms for data
generation and analysis due to various concerns, such as higher
degrees of concurrency, deeper memory hierarchies, substantial I/O and
communication constraints. Simulations today typically save all data
to analyze later. Simulations at the exascale will require us to
analyze data as it is generated and save only what is really needed
for analysis, which must be performed predominately \textit{in-situ},
i.e., executed sufficiently fast locally, limiting memory and disk
usage, and avoiding the need to move large data across nodes.

In this paper, we present a distributed method that enables 
\textit{in-situ} data analysis for large protein folding trajectory
datasets. Traditional trajectory analysis methods currently follow a
centralized approach that moves the trajectory datasets to a
centralized node and processes the data only after simulations have
been completed. Our method, on the other hand, captures conformational
information \textit{in-situ} using local data only while reducing the
storage space needed for the part of the trajectory under
consideration. This method processes the input trajectory data in one
pass, breaks from the centralized approach of traditional analysis,
avoids the movement of trajectory data, and still builds the global
knowledge on the formation of individual $\alpha$-helices or
$\beta$-strands as trajectory frames are generated.
\end{abstract}

\begin{wileykeywords}
exascale computing, I/O reduction, ensemble protein folding, eigenvalues
\end{wileykeywords}

\clearpage


\begin{figure}[h]
\centering
\colorbox{background-color}{
\fbox{
\begin{minipage}{.5\textwidth}
\includegraphics[width=80mm,height=40mm]{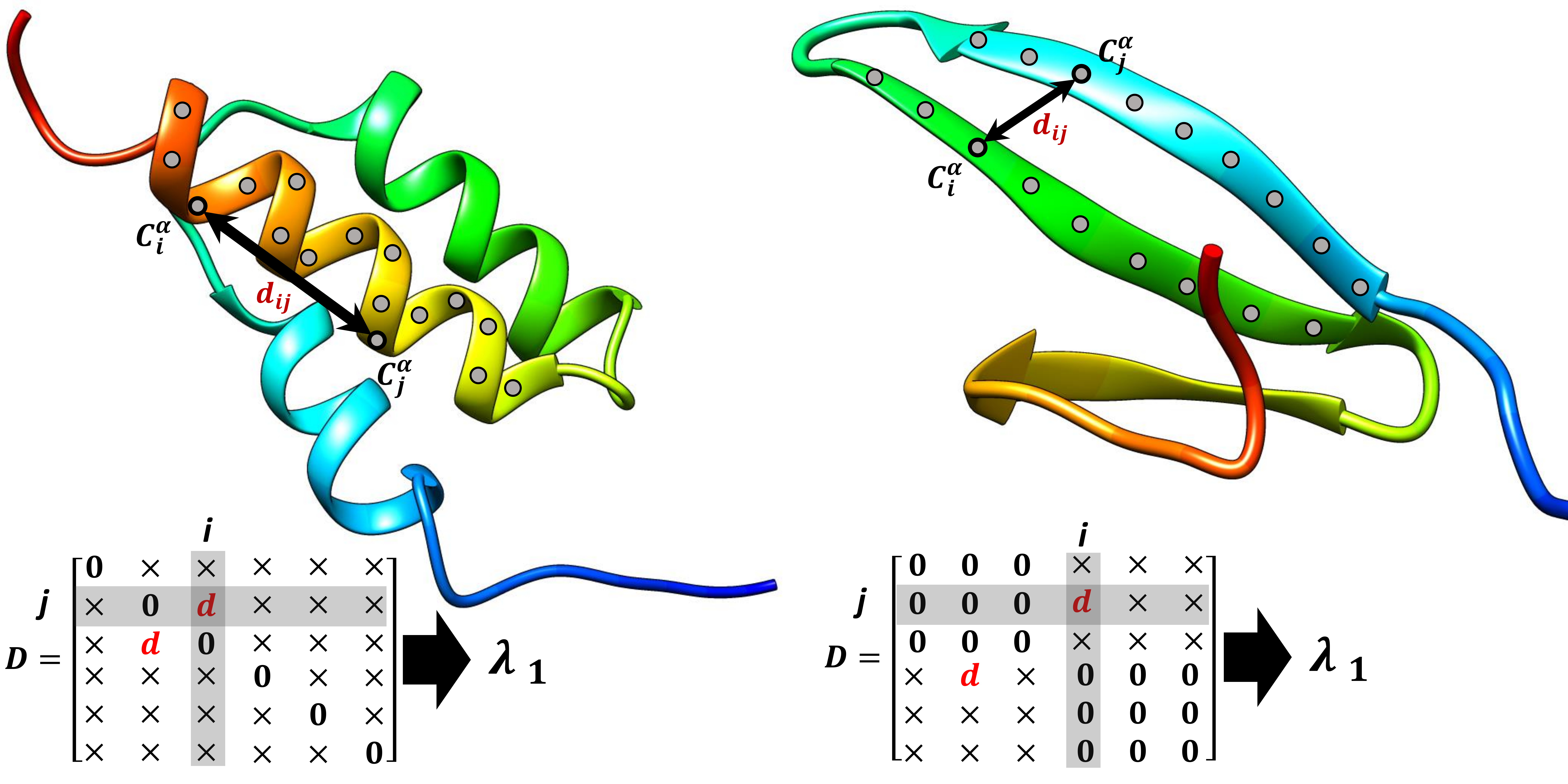}
\end{minipage} \hfill
\begin{minipage}{.5\textwidth}
As computing moves towards exascale, I/O bandwidth limitations and power concerns
will require a fundamental change in the way data is analyzed and stored.
We propose a novel method of \textit{in-situ} data analysis of protein folding trajectories.
The analysis runs in parallel to the simulation and can dramatically reduce the amount
of data sent across a network and written to disk.  
We empirically demonstrate the effectiveness of our approach by investigating
trajectories of two proteins: 1BDD (left) and 1E0L (right).
\end{minipage}
}}
\end{figure}


\makeatletter
\renewcommand\@biblabel[1]{#1.}
\makeatother

\bibliographystyle{apsrev}
\renewcommand{\baselinestretch}{1.5}
\normalsize

\clearpage

\section*{\sffamily \Large INTRODUCTION} 

Fundamental changes in computer architecture will accompany the
transition from petascale to exascale computing.
Aspects like significantly higher concurrency, deeper memory
hierarchies, substantial I/O and communication bottlenecks, and power
constraints are forcing the community to revisit traditional algorithms.
The higher degree of concurrency speeds the generation of
simulation data but I/O and communication bottlenecks, as well as power
constraints, severely limit data storage and movement.
While simulations today save all the data to be analyzed later,
simulations at the exascale will require us to analyze data as it is
generated and save only what is really needed.  New techniques are
needed for the data analysis.  When integrated in simulations, the
analysis should be performed \textit{in-situ},i.e. execute sufficiently
fast locally, use a small amount of memory and disk, and avoid
substantial or frequent data movement~\cite{bennett_2012}.
In this paper we address the challenge of defining an \textit{in-situ}
data analysis for specific datasets: ensembles of trajectories
in protein folding simulations. In past work we addressed a similar
challenge for protein ligand docking 
simulations~\cite{estrada2012, boyu_2013, boyu_2015}.  This work is an expansion and
refinement of work on smaller proteins~\cite{boyu_2014};
in particular, the method we present here is suitable for much larger
proteins than in the previous work and provides a finer-grained,
multi-faceted view of the protein conformations as they evolve through
a simulation.

In protein folding simulations, a protein's string of amino acids
connected by peptide bonds folds into a compact shape, called tertiary
or native structure, which determines how the protein functions. An
intermediate level between primary and tertiary structure is called
secondary structure and it is composed of $\alpha$-helices,
$\beta$-strands, and the turns and loops that connect
them. $\alpha$-helices and $\beta$-strands involve well-understood local
interactions via hydrogen bonding between neighboring amino
acids. $\beta$-strands align and bond with other strands forming
$\beta$-sheets. Molecular dynamics (MD) simulations aim to understand
the properties of these dynamic systems at the atomic level by
providing detailed information on the fluctuations and conformational
changes as they fold.

The increase of computing power and the high degree of parallelism in
the folding simulations enable unprecedented fine-grained time scales,
and allow an increasingly large number of trajectories to be computed
in parallel. The analysis of these trajectories, on the other hand,
is highly centralized, requiring major data movements that will no
longer be sustainable on exascale machines. To reach any conclusion,
traditional analyses initially move all the frame data to a parallel file
system (e.g., Lustre or GPFS). Data are eventually moved to the user's
machine or a cluster dedicated to the analysis, requiring a second
massive movement of data (see Figure~\ref{fig:fig1overview}).

The method we propose in this paper enables an \textit{in-situ} analysis
for large distributed trajectory datasets and removes the need for moving
large amounts of data, thus better fitting with the profile of
exascale machines.  As shown in Figure~\ref{fig:fig2overview}, our
method processes each frame locally and in isolation, transforms each
protein conformation into metadata that is sensitive to conformational
changes, and avoids the movement of trajectory data.
This is a fundamental break from the centralized approach of
traditional data analysis that still builds a global knowledge of
the folding trajectory.

Specifically, as soon as a frame is generated, the frame is written to
local storage and the metadata is generated and stored locally for the
entire length of the simulation.  Stable states are identified by regions
of the trajectory where the metadata changes very little
(i.e., they fluctuate around an average value).
During these periods, one or several representative frames can be
selected and moved from local storage to the parallel file system.
The other frames that are similar can be removed from local storage
without being written (moved) to the parallel file system.
As new frames are generated, each frame's metadata is computed and
compared against the stored metadata.  When new stable regions are
detected, the process repeats: representative frames are
chosen and moved from local storage to the persistent parallel file
system. Transition states of the trajectory can be also identified as
the metadata drifts away from average values; if these transition
states are of interest to the researcher they can be moved to the 
parallel file system. By doing this, our method limits the amount
of data stored on the local hard drive while maintaining the ability to
compare new frames to old, since all metadata on local storage are
preserved. As we will discuss in this paper, the map to metadata and
comparison of metadata is very inexpensive computationally. Thus, we
have a true \textit{in-situ} data analysis which accomplishes our
primary goal.

\section*{\sffamily \Large METHODOLOGY}

Our goal is to identify how trajectories, in ensemble folding
simulations, evolve without the need for writing the entire
trajectory to disk.
We accomplish our goal by independently mapping each conformation
to compact metadata. 
In order to be useful, the map to metadata must preserve
conformational \textit{closeness.}  In other words, similar protein
conformations must map to similar metadata.  The metadata can be
generated and analyzed \textit{in-situ}.  The result of the analysis
determines which frames are ultimately written to disk for access
after the completion of the simulation.

\subsection*{\sffamily \large Mapping a conformation to metadata}

The map to metadata consists of three steps and is applied
independently to each frame of the trajectory.

First, we discretize the protein by selecting (up to) two
representative atoms per amino acid: the $\alpha$-carbon, which is the
carbon in the amino acid backbone, and the $\beta$-carbon, which is the
first carbon atom in the side chain.
The amino acid Glycine is exceptional in that it contains no
$\beta$-carbon; in this case, we
select only one representative for that amino acid. 
The PDB files we used employed a coarse-grained representation;
the CB entry in these PDB files gives the position of the side
chain center. Throughout this paper we will refer to these entries
as $\beta$-carbons, but the reader should keep in mind that these are
actually side chain centers.  We expect that in practice this choice
makes little difference.
We extract the position of each representative carbon as $\vec{r}_1,
\vec{r}_2,...,\vec{r}_{n}$ where $\vec{r}_{i}=(x_i, y_i, z_i)$ is the
Cartesian coordinate of the $i^{\text{th}}$ representative carbon.
Throughout the remainder of this section, we refer to $n$ as the
number of carbons selected in the discretization process.  Note that
$n$ is approximately twice the number of amino acids and is
significantly smaller than the total number of atoms in the entire
protein.

Second, using the selected carbons' atomic coordinates, we create an
$n\times n$ distance matrix, $D=[d_{ij}]$, where $d_{ij}$
records the square of the Euclidean distance between two carbon atoms,
i.e. $d_{ij}=\Vert \vec{r}_i-\vec{r}_j\Vert_{2}^{2}
=(x_i-x_j)^2+(y_i-y_j)^2+(z_i-z_j)^2$. 
The matrix $D$ is called a Euclidean distance matrix, or EDM.

Third, we compute the eigenvalues of the matrix $D$. The matrix $D$
has three important properties that affect the eigenvalues.  $D$ is a
symmetric, real matrix with non-negative entries and zeros along the
diagonal.  Because $D$ is symmetric all of its eigenvalues are real.
Since $D$ is real and symmetric it is diagonalizable; this implies
that $D$ is similar to a diagonal matrix $E$ containing the
eigenvalues of $D$.  The trace of $D$ is clearly zero and the trace is
similarity invariant; hence, the sum of the eigenvalues of $D$ is
zero.  Finally, and most notably, because of the choice of distance
measure, if $n\geq 5$ and all of the carbon atoms do not lie on a
common plane, then $D$ has exactly $5$ non-zero eigenvalues.  We use
non-zero eigenvalues of $D$ to create a metadata packet associated with
each frame.

While we could use all 5 non-zero eigenvalues, we focus exclusively on
the largest eigenvalue.  The reason for this is two fold.
First, there is natural dependence among the eigenvalues.
To see this, suppose that the non-zero eigenvalues of $D$ are
$\lambda_5\leq \lambda_4\leq ...\leq \lambda_{1}$.
Since the sum of the eigenvalues is $0$ we can write:
$\lambda_{5}+\lambda_{4}+\lambda_{3}+\lambda_{2}=-\lambda_{1}$.
So, knowing 4 of the 5 eigenvalues is enough to determine the fifth.
Second, our empirical observations show that the three non-zero eigenvalues
with smallest magnitude are orders of magnitude smaller than the
largest eigenvalue.  We observe that the $5$ non-zero eigenvalues satisfy:
$\lambda_{5}<\lambda_{4}<\lambda_{3}<\lambda_{2}<0 <\lambda_{1}$,
and $|\lambda_{i}|<<\lambda_{1}$ for $i=2, 3, 4$.
This means that $-\lambda_{5}\approx \lambda_{1}$.
The result is that changes in $\lambda_{i}$ for $i=2,3,4$ are nearly
imperceptible and $\lambda_{5}$ is highly correlated with $\lambda_{1}$.

We compute the eigenvalues of $D$ by using standard libraries. Our
code is written in Python and uses the linear algebra libraries
available in scipy (SCIentific PYthon). We use the function {\tt
scipy.linalg.eigh} which takes advantage of the symmetry of the
matrix.

The method described above generates a piece of metadata representing
the entire protein. In practice, it behaves as a very coarse view of
the protein.  In order to get finer grained detail and capture the
formation of an individual $\alpha$-helix or $\beta$-strand, we modify
the method to \textit{zoom in} on a region of interest.  Suppose we
want to track the formation of an $\alpha$-helix (or $\beta$-strand)
and we know that carbons $i$ through $j$ fold into an $\alpha$-helix
(or $\beta$-strand). To focus on the individual structure, we consider
an EDM, $D_{1}$, which is a submatrix of $D$ corresponding to rows and
columns $i$ through $j$. One can equivalently view $D_{1}$ as the EDM
generated when one only considers the representative carbons that form
the substructure.  The matrix $D_{1}$ has the same properties of the
matrix $D$ (i.e., symmetry, non-negativity) but has fewer rows and
columns. Provided that $j-i+1\geq 5$ and not all the coordinates
represented in $D_{1}$ are on a common plane, then $D_{1}$ also has
exactly $5$ non-zero eigenvalues. These eigenvalues share the same
properties as the eigenvalues of $D$; hence we keep only the largest
eigenvalue of $D_{1}$.

To study the relative shape and position of two substructures (e.g.,
two $\alpha$-helices or $\beta$-strands), we build a different matrix
that is derived from $D$, but is no longer an EDM.  Suppose that we
have two $\beta$-strands which are forming and we want to track how
they position themselves relative to each other while forming a
$\beta$-sheet. The first $\beta$-strand includes carbons
$a_1,a_2,...,a_k$ and the second $\beta$-strand includes carbons
$b_1,b_2,...,b_\ell$.  We begin by constructing a (non-square) matrix
$C:=[c_{ij}]$ where
\[c_{ij}:=
(x_{a_i}-x_{b_j})^2 + (y_{a_i}-y_{b_j})^2 + (z_{a_i}-z_{b_j})^2.\]
Then, the inter-structure \textit{distance} matrix we create is the
following:
\[D_{2}:=\left[\text{\begin{tabular}{c|c} $0$ & $C$ \\
\hline $C^{T}$ & $0$ \end{tabular}}\right].\]
The matrix $D_{2}$ is a submatrix of $D$ with rows and columns
$a_1,...,a_{k}, b_1, ..., b_{\ell}$, and with entries $d_{a_i a_j}$
and $d_{b_i b_j}$ zeroed out. The entries $d_{a_i a_j}$ are zeroed out
to remove as much variation as possible from within a single structure.

The matrix $D_2$ is not a EDM because there are zeros off the
diagonal.  However, it is has all the essential properties of an EDM.
$D_{2}$ is real and symmetric which ensures that we have all real
eigenvalues; it is also diagonalizable and $\text{tr}(D_{2})=0$.
Furthermore, because of the block structure of $D_2$ whenever
$\lambda$ is an eigenvalue of $D_2$ so is $-\lambda$.  With this
observation in mind we only consider the positive eigenvalues of
$D_2$.  Additionally, one can prove that $D_{2}$ has at most $5$
positive eigenvalues; for the interested reader, the proof is contained
in the appendix. Again, we select only the largest eigenvalue because
the smaller eigenvalues tend to be much, much smaller.

Throughout the remainder of this paper we use the following naming
convention when referring to matrices. $D$ always refers to an EDM
constructed using all the $\alpha$- and $\beta$-carbons from the entire
protein.  $D_{i}$ always refers to an EDM constructed using all the
$\alpha$- and $\beta$-carbons from the $i^{\text{th}}$ substructure
($\alpha$-helix, or $\beta$-strand). $D_{ij}$ always refers to the
block matrix we last discussed which compares the $i^{\text{th}}$
substructure's relative position to the $j^{\text{th}}$ substructure.

\subsection*{\sffamily \large Eigenvalues preserve conformational
closeness}

In the outset of this section, we asserted any choice of a map to
metadata must preserve computational closeness in order to be useful.
To be useful, we must be certain that if conformations are similar
their metadata will also be similar; we also need to be certain that
when we see disparate metadata that the originating conformations
are also quite disparate. The eigenvalues of the matrices
$D$, $D_{i}$, and $D_{ij}$ give us this guarantee.
This is a result of the stability of eigenvalues~\cite{tao}.
The eigenvalues of a matrix are said to be stable if small perturbations
in the matrix result in only small perturbations in the eigenvalues.
Not all matrices have stable eigenvalues; however, real, symmetric
matrices do have stable eigenvalues and we always consider real
symmetric matrices.  The change in conformation from one frame to
another can be viewed as a perturbation.  The second frame can be
viewed as the first frame with some changes (perturbations) applied.
If the conformations are very similar, then the perturbation is very
small; as a consequence, the eigenvalues must, by stability, be close
together.  If the conformations are very different, then the 
perturbation will be very large.  When the perturbation is large the
change difference in eigenvalues can also be large.
In principle, there is no guarantee that the difference in these
eigenvalues will be large; however, in our experiments we did not
observe cases where the conformations were quite different and the
eigenvalues were quite similar.

\clearpage

\section*{\sffamily \Large RESULTS}

In this section we present empirical results that demonstrate the
accuracy and efficiency of our method for \textit{in-situ} data analysis
of trajectories.

\subsection*{\sffamily\large Datasets}

We consider two proteins: 1BDD and 1E0L. The native structures of these
proteins were experimentally determined by NMR spectroscopy.
The protein 1BDD consists of 60 amino acids and it is composed of
a bundle of three $\alpha$-helices as shown in Figure~\ref{fig:1bdd}.
Helix~2 (Glu25-Asp37) and Helix~3 (Ser42-Ala55) are antiparallel
to each other, and Helix~1 (Gln10-His19) is tilted with respect to the
other two~\cite{1bdd}. The protein 1E0L consists of 35 amino acids
and it presents a triple-stranded antiparallel $\beta$-sheet 
topology~\cite{1e0l} as shown in Figure~\ref{fig:1e0l}.

Molecular dynamics simulations for 1BDD and 1E0L were carried out with
the coarse-grained UNRES force
field~\cite{liwo_2008,liwo_2011,baranowski_2015} using the
parameterization for $\alpha$-helical~\cite{liwo_2007} and
$\beta$-sheet proteins~\cite{liwo_2008_02}, respectively. The
simulations were run at $T=300$~K. About 100,000,000 MD steps were run
at the time step of 4.89 fs. Frames were collected every 5,000 steps.
Langevin dynamics were applied scaling the water friction coefficient by
a factor of 100 as in our earlier work~\cite{liwo_2005,khalili_2005_02}.

\subsection*{\sffamily \large Formation of individual structures}

We first assess the ability of our method to identify the formation
of individual secondary structures (i.e., a single $\alpha$-helix or
a single $\beta$-strand). We compute the eigenvalues for the region
of the protein that eventually forms the secondary structure being
examined; these eigenvalues are associated to distance
matrices of type $D_{i}$ referenced in the methodology section.

We begin our study by examining a folding trajectory of the protein
1BDD which contains a bundle of 3 $\alpha$-helices.
The regions which fold into helices span 10, 13, and 14 amino acids
respectively. The total number of $\alpha$- and $\beta$-carbons in each
$\alpha$-helix is 20, 25, and 28 respectively.
Thus, in the process of computing the eigenvalues for each folding
helix, we consider three distance matrices per frame.
Each matrix is relatively small: $20\times 20$, $25\times 25$,
and $28\times 28$ respectively.  
Note that each matrix easily fits in memory. 
We compute the largest eigenvalue of each of these matrices and
store it as metadata.

Figure~\ref{helix2_formation} shows to the formation of Helix~2 in
the protein 1BDD. Figure~\ref{helix2_formation}.a (top) shows the
evolution of the largest eigenvalue for this helix over the first 100
frames of the trajectory. Each point in the plot corresponds to the
largest eigenvalue of the distance matrix $D_{2}$ for a given frame.
Figure~\ref{helix2_formation}.b (middle) shows the root-mean-squared
deviation (RMSD) of the structure as compared to the corresponding
native structure which had been determined using NMR spectroscopy.
Looking at the largest eigenvalue sequence we see three stages.
The first stage (blue) has the largest values and the most variability.
The second stage (red) has much less variability.
The final stage (green) is made up of smaller values and has less
variability than the previous stages.  
We selected two frames from each region and visualized the
corresponding conformation using Chimera~\cite{chimera} in
Figure~\ref{helix2_formation}.c (bottom).  
The highlighted (cyan) portion of the protein corresponds to the
region under consideration that eventually becomes Helix~2.
We observe that in the first stage, the helix is completely unfolded.
Throughout the second stage, the helix has begun to form and contains
a single helical turn. By the third stage, the helix is much better
formed and consists of several helical turns.

The construction of Helices 1 and 3 follows a similar behavior but,
because one of their ends is loose, the substructures exhibit more
variability and less stability. 
Figure~\ref{helix3_formation} presents results related to the
protein region that folds into Helix~3; the structure of the figure
is similar to Figure~\ref{helix2_formation}:
Figure~\ref{helix3_formation}.a (top) and 
Figure~\ref{helix3_formation}.b (middle) shows the largest eigenvalue
and RMSD patterns and Figure~\ref{helix3_formation}.c (bottom) shows
representative snapshots. The third helix forms very rapidly.
We observed this formation in both the eigenvalue and RMSD patterns.
We also see only two stages: the
first stage before frame 67 and the second after frame 67. The sudden
change in the eigenvalue pattern coincides with the sudden formation
of the helix. In Figure~\ref{helix3_formation}.b we visualize three
frames before and three frames after frame 67.

We observed that the trajectory of eigenvalues has a remarkably
similar trajectory to the RMSD (comparing the shapes in
Figure~\ref{helix2_formation}.a to Figure~\ref{helix2_formation}.b
and comparing Figure~\ref{helix3_formation}.a 
to Figure~\ref{helix3_formation}.b).
The important thing to note is that, in practice,
RMSD values are measured against a known reference structure that is
not always available. The computation of the eigenvalues, on the other
hand, is performed on a single frame in isolation (i.e., without
looking at the previous or next frames).

The all-$\beta$ protein 1E0L is composed of three $\beta$-strands. The
regions which fold into strands span 6, 7, and 5 amino acids
respectively.  The total number of $\alpha$- and $\beta$-carbon atoms
considered in each strand is 12, 14, and 10 respectively. When we
consider the distance submatrix used for the computation of
eigenvalues, their sizes are $12\times 12$, $14\times 14$, and
$10\times 10$. Figure~\ref{sheet1_validation},
Figure~\ref{sheet2_validation}, and Figure~\ref{sheet3_validation}
refer to the first, second, and third $\beta$-strands respectively.
Figure~\ref{sheet1_validation}.a (top) shows the pattern of
the largest eigenvalues obtained from the distance matrices composed 
of the $\alpha$- and $\beta$-carbon atoms of the first $\beta$-strand--one
eigenvalue for each frame of the trajectory.
We cluster the pattern in regions; 
each region's values appeared to be stable. 
We hypothesize that these states are associated with meta-stable strands
in the trajectory.
From each cluster, three representatives are chosen by selecting frames
whose $\beta$-strand's eigenvalues are nearest the mean value of the
eigenvalues in the region.
The six sets of frames are visualized in Figure~\ref{sheet1_validation}.b
(bottom) using Chimera;
frames that belong to the same cluster are grouped and boxed using the
same color from the eigenvalue plot in 
Figure~\ref{sheet1_validation}.a (top).
In Figure~\ref{sheet1_validation} we see that within a cluster the
highlighted region of the protein are mostly similar.
The fourth cluster (blue) is a bit of an exception and is not as
internally similar as the rest.
Also, we notice that the second and fifth clusters (green and purple
respectively) have both similar eigenvalues and a similar appearance.

Figure~\ref{sheet2_validation} and Figure~\ref{sheet3_validation} are
structured similarly as Figure~\ref{sheet1_validation} but refer to
the second and third $\beta$-strands respectively. The top of each figure
presents the clustering of the largest eigenvalues for the considered
$\beta$-strand; on the bottom we present the structures of three frames
for each cluster that are selected because the $\beta$-strand's eigenvalues
are nearest the mean value of the eigenvalues in the associated cluster.
Following the pattern for visualizing the first $\beta$-strand,
we identify six clusters for the second strand and four clusters for
the third strand. 
Again, we observe similarities among structures that are in the same
cluster and have similar eigenvalues.

The visual analysis of similarities previously remarked in
Figure~\ref{sheet1_validation}, Figure~\ref{sheet2_validation},
and Figure~\ref{sheet3_validation} is purely qualitative.
In order to move from a qualitative and visual approach to a 
quantitative method to assess the similarities,
we construct the heat maps for the three $\beta$-strands.
The heat map is a visual representation of the RMSD of pairs of
representative frames.
Figure~\ref{sheet1_heatmap} shows the heat map for the first strand,
Figure~\ref{sheet2_heatmap} for the second, and
Figure~\ref{sheet3_heatmap} for the third.  In these figures we
compute the RMSD from the $i^{\text{th}}$ representative $\beta$-strand
to the $j^{\text{th}}$ representative $\beta$-strand.
The RMSD value is stored in the $ij$ entry of a matrix.
The matrix is visualized as a heat map where lighter (whiter) areas
correspond to lower RMSD values (i.e., white being an RMSD of 0 Angstroms)
and darker (redder) areas correspond to higher RMSD values.
The heat map is blocked into regions by cluster.

Based on the conformation visualization in
Figure~\ref{sheet1_validation} we expect the $3\times 3$ block
diagonal portion of Figure~\ref{sheet1_heatmap} to be very lightly
colored. We indeed observe this phenomena.  The lightly colored block
diagonal implies that conformations within a cluster have small RMSD and
are therefore conformationally similar.
Additionally, we confirm a visual observation about the second and
fifth clusters being similar.  This observation is reflected in the
heat map in which we note that the blocks corresponding to
Cluster 2 and Cluster 5 are nearly as light as the block diagonal.
The other important feature of this heat
map is that the $3\times 3$ blocks on the super-diagonal are quite
dark colored. This indicates that when the eigenvalues significantly
change resulting in a new cluster, the associated $\beta$-strand also
significantly changes. Some of the most profound changes are observed
between Clusters 1 and 2, and Clusters 3 and 4.

Figure~\ref{sheet2_heatmap} corresponds to the selected
representatives of the second $\beta$-strand in
Figure~\ref{sheet2_validation}. For this strand, we observe many of
the same features in the heat map as the visual analysis. Once again,
the block diagonal is very light, as expected. Furthermore, we see
from Figure~\ref{sheet2_validation} that the third and fourth clusters
are visually similar as are the second and fifth. These observations
are confirmed with the heat map.

Figure~\ref{sheet3_heatmap} is the heat map corresponding to the
representatives of the third $\beta$-strand in Figure~\ref{sheet3_validation}.
We note one final observation: the eigenvalues from each strand settle
into a final state at approximately the same time near frame
32000 in the simulation.  This is likely a result of the interdependence
of the strands. A single strand, by itself, is not very rigid; 
however, when the two strands are held together by hydrogen bonds,
they are able to form a more rigid structure.

\subsection*{\sffamily \large Relative position of two substructures}
	
In this section we assess the capability of our method to capture the
relative position of a protein's secondary structures with respect to
each other over the trajectory. Our empirical study considers once
again an entire trajectory of the 1BDD protein with its three
$\alpha$-helices.  
As a first view, we consider the largest eigenvalue of the distance
matrix, $D$, formed by considering all the $\alpha$- and $\beta$-carbons.
Then, to refine our view of the protein we consider the largest
eigenvalues from the distance matrices for the individual helices
(i.e., Helix~1, Helix~2, and Helix~3); these are the matrices of type
$D_{i}$ referenced in the methodology section.
We also consider the largest eigenvalues from the matrices which
compare distances between two helices at a time
(i.e., Helix~1 with Helix~2, Helix~1 with Helix~3, and Helix~2 with
Helix~3); these are the matrices of type $D_{ij}$ referenced in the
methodology section.

Examining the pattern of the largest eigenvalues for the entire
protein in Figure~\ref{ev_helix3_moving}.a (top) we observe a spike
between Frames 1300 and 1400.  When the spike occurs, the eigenvalue's
magnitude is tripled. If we only considered this coarse grained view,
all we can conclude is that during those frames, a structural change
in the protein conformation happens but we are unable to pinpoint
the precise cause.

To investigate the cause of the spike, we examined the three
eigenvalue patterns associated to the distance matrices corresponding
to an individual helix substructure. If there were a corresponding change
in the eigenvalues, we could, for example, surmise that one or more of
the helices may have unfolded and refolded again.
Figure~\ref{ev_helix3_moving}.b (middle) contains plots of these
eigenvalues over the region of interest.
We note that there are no significant changes in any of the eigenvalues
over these ranges; thus we are led to conclude that the helix structure 
remained relatively unchanged over this portion of the trajectory.
Finally, we consider the eigenvalues of the inter-structure distance
matrices.  When the relative positions of two rigid structures is
constant there should be very little change in these matrices and hence,
very little change in the corresponding eigenvalues.
However, if two structures are moving either apart or together,
we expect to see some significant changes in the associated distance
matrix which manifests itself as significant changes in the eigenvalues.
Figure~\ref{ev_helix3_moving}.c (bottom) shows the largest eigenvalue
of these inter-structure distance matrices.
First, we notice that the relative distance between Helix~1 and Helix~2
appears to change very little.
On the other hand, we quickly observe that there is a
prominent spike in the eigenvalues of the other two matrices comparing
Helix~1 to Helix~3 and comparing Helix~2 to Helix~3.
The common structure represented by those two matrices is Helix~3.
We hypothesize that the explanation for the spike observed in
the coarse view of the entire protein is a result of the third helix
moving drastically with respect to the first and second helix.  This
hypothesis is confirmed when we view the conformations from the
trajectory. Figure~\ref{helix3_moving} shows consecutive frames
between frames 1300 and 1400; Helix~3 (orange) swings away
from Helix~1 and Helix~2 (cyan and magenta respectively) and
then returns to its initial position. Note that the colors of the
helices in this figure correspond to the coloring of the eigenvalues
in Figure~\ref{ev_helix3_moving}.b (middle).
Over the course of this transformation, the changes in the relative
position of the first two helices are small.

From these two case studies, of a 1BDD folding trajectory and a 
1E0L folding trajectory, we observed two major strengths of our method.
First, we were able to positively identify stable stages
of both trajectories.  We observed that during periods where the
eigenvalues were stable the conformation was also stable.  This
bolsters our confidence in the ability to use distance between eigenvalues
as an accurate proxy for the distance between conformations (typically
measured with RMSD or similar metric).  The most striking observation
was the similarity in shape between RMSD and the largest eigenvalues
that we saw in Figures~\ref{helix2_formation}.a-b and 
Figures~\ref{helix3_formation}.a-b.
Second, we were able to leverage our ability to \textit{zoom-in}
on segments of the protein, and to compare two segments, to better
understand the clues we saw at a coarse scale.  The ability to
view the protein from multiple angles proved useful to analyze the
cause of the spike we saw in Figure~\ref{ev_helix3_moving}.

\clearpage
\section*{\sffamily \Large DISCUSSION}

In this section we verify that our analysis is indeed an
\textit{in-situ} data analysis.  We discuss the novelty
of our approach and compare it with popular existing approaches.
We conclude the section with a discussion of the challenges and
opportunities related to integrating our method into existing
tools for simulation and analysis.

\subsection*{\sffamily \large The cost of our \textit{in-situ} 
data analysis}
	
In order to be considered an \textit{in-situ} data analysis,
the CPU and memory footprint must be light enough to not interfere
with the ongoing simulation and communication must be minimal.
We note that our method does not communicate any information across
the network, so its communication footprint is nonexistent.

Beginning with the trajectory for 1E0L (39130 frames, 35 amino acids)
we construct 7 matrices for each frame (i.e., one matrix for the entire
protein, one matrix for each of the three $\beta$-strands, and one matrix
for each of the pairs of $\beta$-strands).
Using a single CPU core running at 2.66 Ghz, it took 768.059s to
compute all the eigenvalues;
this implies an average CPU time of 0.0196s per frame.
During a typical simulation, frames are generated using many cores of
the node at a rate of about 6 seconds per frame;
the 0.0196s on a single CPU is a negligible fraction of this time.
The process of computing the eigenvalues required approximately 100 KB
of memory.
Finally, storing in memory the largest eigenvalue from each of the
7 matrices for the entire length of the trajectory requires about
2.2 MB memory.

The footprint is similarly light for the trajectory of 1BDD 
(41896 frames, 60 amino acids).
The construction of the 7 matrices for the entire trajectory required
1417.651s implying an average CPU time of .0338s per frame.
The matrices are somewhat larger and therefore require a little
additional memory. The computation used approximately 200 KB of memory. 
Storing the largest eigenvalue (for 1BDD or 1E0L) requires the same
amount of memory per frame, $7\times 8=56$ Bytes per frame, or about
2.35 MB for the entire trajectory.  

We see that, both memory and CPU utilization are minimal with the
proteins we considered.  Thus, our analysis is an \textit{in-situ}
data analysis.  Proteins with more substructures of interest require
storing more eigenvalues.  For example, if the protein has 10
substructures, then there are up to $56$ eigenvalues per frame (i.e.,
1 eigenvalue for the matrix $D$, 10 matrices of the form $D_{i}$
contribute 1 eigenvalue each, and 45 matrices of the form $D_{ij}$),
or a total of about 18 MB.  This is still an insignificant amount of
memory to be used.  A further memory reduction
can be achieved by storing the eigenvalues as single precision
floating point numbers instead of double precision; this cuts the
memory requirement in half with negligible loss of accuracy.

\subsection*{\sffamily \large The novelty of our approach}

There are a number of methods that have been used to analyze trajectory
data.  One method, that works in a parallel distributed fashion is a
framework called HiMach~\cite{tu}.  This framework is a MapReduce style
interface that takes advantage of naturally parallel analysis operations
to understand statistical data of long trajectories.
Our work differs in that we focus on the similarity or discrepancy in
the geometric structure of the conformation and not on any statistical
information about the frame.

More sophisticated trajectory analyses, like those of 
Best et al.~\cite{best} and Phillips et al.~\cite{phillips} are
traditionally centralized, and are therefore limited by the length
of the trajectory and the size of the protein.  These traditional
methods also make comparisons to an energy minimal structure known
ahead of time.  Their works focus on constructing a
frame-by-frame dissimilarity matrix and making reductions to lower
dimensionality.
Our work differs in that our analysis is not centralized and requires
no data movement; our analysis is accomplished as the simulation runs
(not post simulation); we require no \textit{a priori} knowledge of
an energy minimal structure; and, we are able to consider much larger
proteins because we focus on smaller substructures.

The idea of mapping to metadata and analyzing the metadata is not new.
Our previous work~\cite{boyu_2015}, which explored protein-ligand docking,
mapped ligand geometries to metadata using a sequence of projections
and regression.  In that work, the orientation of the ligand
(both translation and rotation) is important; the analysis of metadata
is done after the simulation and on the entire set of metadata.
Our present work is different because we analyze the metadata
during the simulation.  In addition, our choice of metadata differs
because the orientation of the protein (translation and rotation) is
irrelevant; we need (and construct) metadata that is rotation and
translation insensitive.

We have previously developed a method to analyze protein folding
trajectories~\cite{boyu_2014}.  Our earlier work involved map to metadata
that is very similar to the work in this paper.
The primary difference is that previously, we computed eigenvalues of
a different matrix (not a Euclidean distance matrix) and we only 
considered a \textit{coarse} view of the protein (building the matrix
from all $\alpha$- and $\beta$-carbons).
We were able to show that the method could identify stable states,
but only in relatively simple proteins.  And, without looking at
matrices built from substructures, we were unable to get a
fine-grained view of what was happening in the trajectory--we
were only able to detect meta-stable and transition states.
Our new approach gives us a much more refined view of the protein
allowing us to extract more information about the conformational 
evolution.

\subsection*{\sffamily \large Opportunities to adopt and expand our method}

We demonstrated that our method effectively captures conformational
changes in proteins and that the metadata generated can be used
to deduce what is happening in the trajectory without the need
for moving the trajectory data and visualizing it.  

In order to fully realize the potential of our method,
it must be used as the simulation is running and the ensemble
of trajectories across the compute nodes evolve.
We envision the incorporation of our method into an existing framework
for \textit{in-situ} data analysis such as DataSpaces~\cite{ds}
and the tools for analysis being incorporated into well-known
toolkits such as MMTSB~\cite{mmtsb}.

The analyses of eigenvalue patterns presented in this paper are
performed manually.  The scalability of our approach can be assured by
extending our method and integrating automatic clustering algorithms
that the classify the trajectory state into either stable or
transitional states based on the eigenvalue pattern.  
There are several popular methods for clustering data including
centroid-based clustering (e.g. fuzzy $c$-means),
density-based clustering (e.g. DBSCAN), and divide-and-conquer
strategies making repeated use of fuzzy $c$-means, and others.
To achieve maximum utility and to scale to large-scale simultaneous on
thousands of nodes, it is necessary to integrate machine learning 
algorithms which intelligently, and automatically, detect the changes
between these two states.

Finally, in this paper we focus entirely on \textit{in-situ} data
analysis.  It is easy to envision transforming our method into an
\textit{in-transit} data analysis to cross compare large ensembles of
trajectories evolving in parallel.  In this case, the eigenvalue
metadata is communicated to a central node for an ensemble analysis.
The \textit{in-transit} analysis is tasked with keeping track of which
trajectories or substructures are rapidly evolving. The filtered
knowledge can be ultimately used by the scientist to tune many 
simultaneous simulations on the fly
(e.g., terminating simulations that quickly converged to a
folded protein).

\clearpage
\section*{\sffamily \Large CONCLUSIONS}

As computing moves towards exascale the concurrency of supercomputers
is increasing dramatically; however, because of power constraints the
I/O bandwidth is essentially unchanged.  The increase in concurrency
makes it possible to run ever larger ensembles of protein folding
simulations.  These large ensemble simulations are capable of
generating data faster than it can be written to disk, causing the
movement of data to disk to substantially slow down the simulations.
Traditional analyses of trajectory data do not scale on exascale
machines because they rely on storing the entire generated datasets to
disk first, and moving the data to a centralized node or
dedicated cluster after the simulation is completed.

Our method of \textit{in-situ} data analysis makes it possible to
scale ensemble protein folding simulations in such a way to take full
advantage of the increasing concurrency of these future machines by
mapping each conformational frame in isolation into one or multiple
eigenvalues. The eigenvalues serve as metadata for the trajectory
analysis. We empirically demonstrate that, as a direct result of
eigenvalue stability, our choice of metadata captures conformational
changes for two different proteins: 1BDD and 1E0L composed of a bundle
of three $\alpha$-helices and a triple-stranded antiparallel
$\beta$-sheet topology respectively. Dealing with eigenvalues and
using their pattern locally to select relevant frames to save to disk,
we can significantly reduce the data movement during simulations. The
eigenvalues can be ultimately used to monitor the transition of
protein structures to stable states on the fly, without the need for
the scientist to move entire trajectories to local disk and analyze
them after the simulations are completed.

\clearpage
\section*{\sffamily \Large ACKNOWLEDGMENTS}

This work is supported by NSF grants \#CCF-1318445/1318417. The
authors gratefully acknowledge the use of Chimera for each of the
visualizations of protein conformations presented in this paper and
for the computation of RMSD. Chimera is developed by the Resource for
Biocomputing, Visualization, and Informatics at the University of
California, San Francisco (supported by NIGMS P41-GM103311). The
authors thank Ms. Agnieszka Lipska, Faculty of Chemistry, University
of Gdansk, for supplying the coarse-grained UNRES trajectories of
protein A and FBP-28 WW domain; Adam Liwo acknowledges the support of
Grant DEC-2013/10/M/ST4/00640 from the Polish National Science Center.

\clearpage
\section*{\sffamily \Large APPENDIX}

\subsection*{\sffamily \large The number of positive eigenvalues
matrices $D_{ij}$}

\begin{Proposition}
The matrix $D_{ij}$ has at most $5$ positive eigenvalues.
\end{Proposition}

\begin{proof}
First, we note that $D_{ij}$ can be viewed as a \textit{transformed}
Euclidean distance matrix.  Specifically, let $D$ be the Euclidean
distance matrix formed using all the selected atoms from the two
structures represented by $D_{ij}$.  In block form, $D$ can be
decomposed as  
\[D=\left[\text{\begin{tabular}{c|c} $A_{ii}$ & $A_{ij}$ \\
\hline $A_{ji}$ & $A_{jj}$ \end{tabular}}\right]\]
where $A_{ii}$ contains distances between two points in structure $i$,
$A_{jj}$ contains distances between two points in structure $j$, and
$A_{ji}=A_{ij}^{T}$ contains distances between one point in structure $i$
and one point in structure $j$.
By construction, $D_{ij}$ has the form:
\[D_{ij}:=\left[\text{\begin{tabular}{c|c} $0$ & $A_{ij}$ \\
\hline $A_{ji}$ & $0$ \end{tabular}}\right].\]
With \textit{slight} abuse of notation, let us say that $A_{ij}$ as
$i$ rows and $j$ columns and hence $D$ and $D_{ij}$ have $i+j$ rows
and columns.  Let $r_{k}$ be the $k^{\text{th}}$ row of $D$ and
let $r_{k}^{\prime}$ be the corresponding $k^{\text{th}}$ 
row of $D_{ij}$.

Because of the block structure of $D_{ij}$ whenever $\lambda$ is
an eigenvalue of $D_{ij}$ so is $-\lambda$.  Thus, to show
that $D_{ij}$ has at most 5 positive eigenvalues, it suffices to
show that $D_{ij}$ has at most 10 non-zero eigenvalues.
We recall that since $D$ is a Euclidean distance matrix of a
$3$-dimensional object that $D$ has at most $5$ non-zero eigenvalues.
Thus, $D$ has rank (at most) $5$ and consequently at most $5$ rows
of $D$ can be linearly independent.  We also note that it is clear
that any row $r^{\prime}_{k}$ with $k>i$ is linearly independent of the
set of rows $\{r^{\prime}_{1}, r^{\prime}_{2},...,r^{\prime}_{i}\}$.

We will now show that at most $5$ rows of
$\{r^{\prime}_{1},...,r^{\prime}_{i}\}$ can be linearly independent.
To do this, we need only show that any 6 of those rows are linearly
dependent.  Without loss of generality, we can assume those rows
are $r^{\prime}_{1}, r^{\prime}_{2},...,r^{\prime}_{6}$.
Since $D$ has rank at most $5$ we know that $r_1, ..., r_6$ are linearly
dependent.  Thus, there exist $\alpha_1,...,\alpha_{6}$ (not all zero)
so that 
\[ \alpha_{1}r_1 + \alpha_2 r_2 + ... + \alpha_{6}r_{6} = 0.\]
It follows that the same choice of coefficients yields:
\[ \alpha_{1}r_{1}^{\prime} + \alpha_{2}r_{2}^{\prime} +
... + \alpha_{6}r_{6}^{\prime} = 0.\]
This is most easily seen when we think about the vectors $r_{k}$ and
$r^{\prime}_{k}$ coordinate-wise.  For coordinates $1,...,i$ the
vectors $r_{k}$ and $r^{\prime}_{k}$ have the same (non-zero) value.
But, because $\alpha_{1}r_1 + \alpha_2 r_2 + ... + \alpha_{6}r_{6} = 0$
it follows that, in the second sum, that coordinate is zero.
In the remaining coordinates $i+1, i+2, ..., i+j$ all the entries of
$r_{k}^{\prime}$ are zero, hence any choice of $\alpha$ would result in 
zeros in these coordinates.  Thus 
$\alpha_{1}r_{1}^{\prime} + \alpha_{2}r_{2}^{\prime} + ... +
\alpha_{6}r_{6}^{\prime} = 0$ as claimed and these 6 rows are linearly
dependent.  By a similar argument, any 6 rows of 
$\{r_{i+1}^{\prime}, ..., r_{i+j}^{\prime}\}$ are linearly dependent.
Thus, $D_{ij}$ has at most 10 linearly independent rows; $D_{ij}$
has rank at most $10$; and, therefore, $D_{ij}$ has at most 10
non-zero eigenvalues.
\end{proof}

\clearpage


\bibliography{biblio}

\clearpage

\begin{figure}
\caption{\label{fig:fig1overview} Traditional centralized data analysis.
A single trajectory is computed on one node which writes all of its data to the
parallel file system. Many hundreds or thousands of trajectories can be computed
in parallel; each node writes all of its data to the parallel file system
simultaneously.}
\end{figure}

\begin{figure}
\caption{\label{fig:fig2overview} {\it In-situ} data analysis in our
  proposed method. Each frame is converted into one or multiple
  eigenvalues in isolation. The eigenvalues pattern is used to select
  the frames to store to disk.}
\end{figure}

\begin{figure}
\caption{\label{fig:1bdd} Native structure of protein 1BDD. The
  protein consists of 60 amino acids and it is composed of a bundle of
  three $\alpha$-helices.  Helix~2 (Glu25-Asp37) and Helix~3
  (Ser42-Ala55) are antiparallel to each other, and Helix~1
  (Gln10-His19) is tilted with respect to the other two.}
\end{figure}

\begin{figure}
\caption{\label{fig:1e0l} Native structure of protein 1E0L. The
  protein consists of 35 amino acids and it presents a triple-stranded
  antiparallel $\beta$-sheet topology.}
\end{figure}

\begin{figure}
\caption{\label{helix2_formation} Formation of the second (middle)
  helix of the protein 1BDD.  The largest eigenvalue of the distance
  matrix corresponding to this substructure is plotted above and, for
  comparison, the RMSD for each frame is shown below.  The RMSD is
  measured against a reference structure; only the region
  corresponding to the helix is considered.  Below are representative
  conformations from each of the colored regions.  The highlighted
  (cyan) portion of the protein corresponds to the region under
  consideration that eventually becomes the second helix.  }
\end{figure}

\begin{figure}
\caption{\label{helix3_formation} Formation of the third helix of the
  protein 1BDD.  The largest eigenvalue of the distance matrix
  corresponding to this substructure is plotted above and, for
  comparison, the RMSD for each frame is shown below.  The RMSD is
  measured against a reference structure; only the region
  corresponding to the helix is considered.  Below are representative
  conformations from each of the colored regions.  The highlighted
  (cyan) portion of the protein corresponds to the region under
  consideration that eventually becomes the second helix.  }
\end{figure}

\begin{figure}
\caption{\label{sheet1_validation} Largest eigenvalue of the distance
  matrix for each frame of the trajectory. The distance matrix
  corresponds to the region that will become the first $\beta$-strand.
  Colored regions, clusters, represent expected \textit{meta stable}
  states.  From each cluster, three representatives are chosen.  Each
  frame chosen as a representative is visualized below the eigenvalue
  plot.  The region of the first $\beta$-strand is highlighted in cyan.
}
\end{figure}

\begin{figure}
\caption{\label{sheet2_validation} Largest eigenvalue of the distance
  matrix for each frame of the trajectory.  The distance matrix
  corresponds to the region that will become the second $\beta$-strand.
  Colored regions, clusters, represent expected \textit{meta stable}
  states.  From each cluster, three representatives are chosen.  Each
  frame chosen as a representative is visualized below the eigenvalue
  plot.  The region of the first $\beta$-strand is highlighted in cyan.
}
\end{figure}

\begin{figure}
\caption{\label{sheet3_validation} Largest eigenvalue of the distance
  matrix for each frame of the trajectory.  The distance matrix
  corresponds to the region that will become the third $\beta$-strand.
  Colored regions, clusters, represent expected \textit{meta stable}
  states.  From each cluster, three representatives are chosen.  Each
  frame chosen as a representative is visualized below the eigenvalue
  plot.  The region of the first $\beta$-strand is highlighted in cyan.
}
\end{figure}

\begin{figure}
\caption{\label{sheet1_heatmap} RMSD for each pair of representatives
  selected in Figure~\ref{sheet1_validation} and record this value in
  a matrix.  The $ij$'th entry is the RMSD between representative $i$
  and representative $j$.  When computing the RMSD, we consider only
  the highlighted region of the protein.  The matrix is colored so
  that lighter (whiter) colors correspond to smaller RMSD and darker
  (redder) colors correspond to larger RMSD.  }
\end{figure}

\begin{figure}
\caption{\label{sheet2_heatmap} RMSD for each pair of representatives
  selected in Figure~\ref{sheet2_validation} and record this value in
  a matrix.  The $ij$'th entry is the RMSD between representative $i$
  and representative $j$.  When computing the RMSD, we consider only
  the highlighted region of the protein.  The matrix is colored so
  that lighter (whiter) colors correspond to smaller RMSD and darker
  (redder) colors correspond to larger RMSD.  }
\end{figure}

\begin{figure}
\caption{\label{sheet3_heatmap} RMSD for each pair of representatives
  selected in Figure~\ref{sheet3_validation} and record this value in
  a matrix.  The $ij$'th entry is the RMSD between representative $i$
  and representative $j$.  When computing the RMSD, we consider only
  the highlighted region of the protein.  The matrix is colored so
  that lighter (whiter) colors correspond to smaller RMSD and darker
  (redder) colors correspond to larger RMSD.  }
\end{figure}

\begin{figure}
\caption{\label{ev_helix3_moving} Largest eigenvalue of all three
  types of matrices we consider. The top plot shows the largest
  eigenvalue of the distance matrix formed from the entire protein
  (1BDD).  The middle row of plots show the largest eigenvalue of the
  distance matrix formed by individual substructures (regions that
  become helices).  The color of the plot corresponds to the similarly
  colored region of the protein in Figure~\ref{helix3_moving}.  The
  bottom row of plots show the largest eigenvalue of the
  inter-structures distance matrix formed by looking at two helices.
  There is a large change in the top plot between frames 1300 and
  1400.  There is no corresponding change in the individual structure
  plots (middle row); but we observe corresponding large changes in
  two of the three inter-structure plots.  }
\end{figure}

\begin{figure}
\caption{\label{helix3_moving} Selected conformations of protein 1BDD
  frames 1300-1410.  We observe that the third (orange) helix
  swings away from the other two helices, and then folds back
  together.  This movement was captured by the eigenvalues recorded in
  Figure \ref{ev_helix3_moving}.  }
\end{figure}



\clearpage

\vspace*{0.1in}   
	\begin{center}
		\includegraphics[width=0.8\columnwidth,keepaspectratio=true]{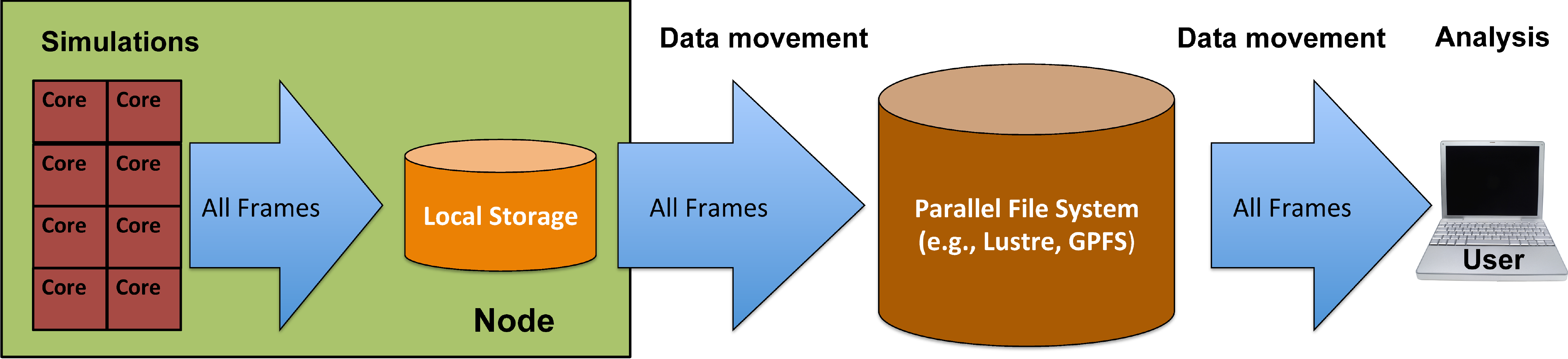}
	\end{center}
\vspace{0.25in}
\hspace*{3in}
{\Large
	\begin{minipage}[t]{3in}
	\baselineskip = .5\baselineskip
	Figure \ref{fig:fig1overview} \\
	Johnston, Zhang, Liwo, Crivelli, and Taufer \\
	J.\ Comput.\ Chem.
	\end{minipage}
}

\vspace*{0.1in}   
	\begin{center}
		\includegraphics[width=0.8\columnwidth,keepaspectratio=true]{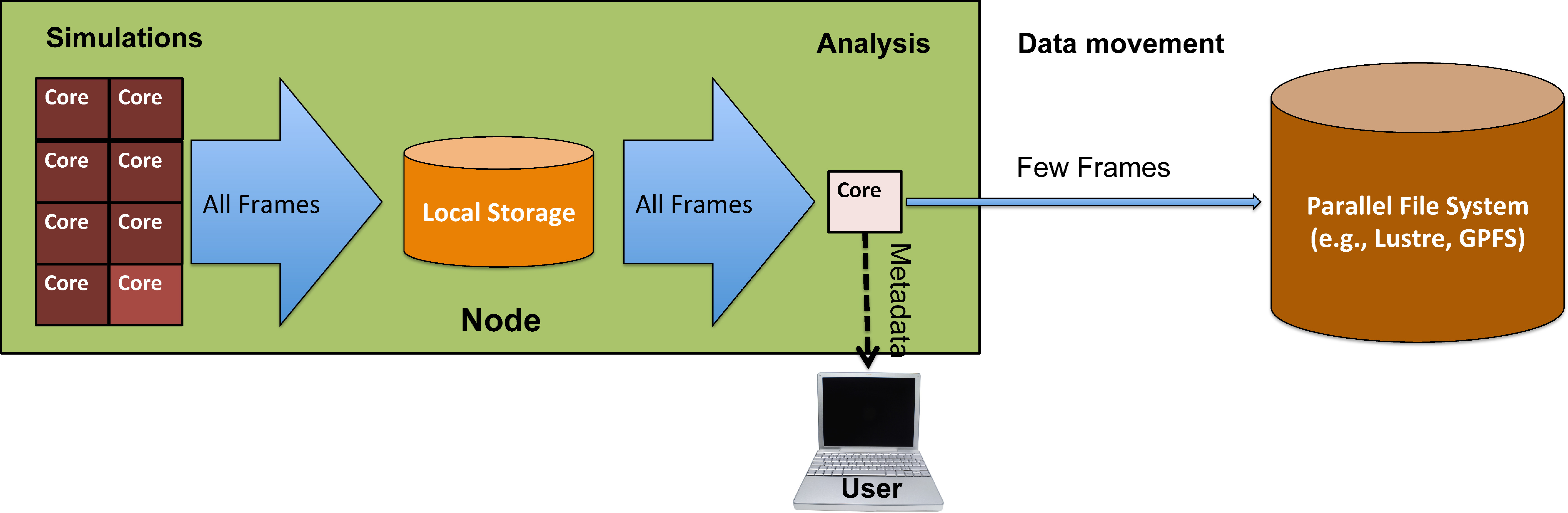}
	\end{center}
\vspace{0.25in}
\hspace*{3in}
{\Large
	\begin{minipage}[t]{3in}
	\baselineskip = .5\baselineskip
	Figure \ref{fig:fig2overview} \\
	Johnston, Zhang, Liwo, Crivelli, and Taufer \\
	J.\ Comput.\ Chem.
	\end{minipage}
}

\vspace*{0.1in}   
	\begin{center}
		\includegraphics[width=0.8\columnwidth,keepaspectratio=true]{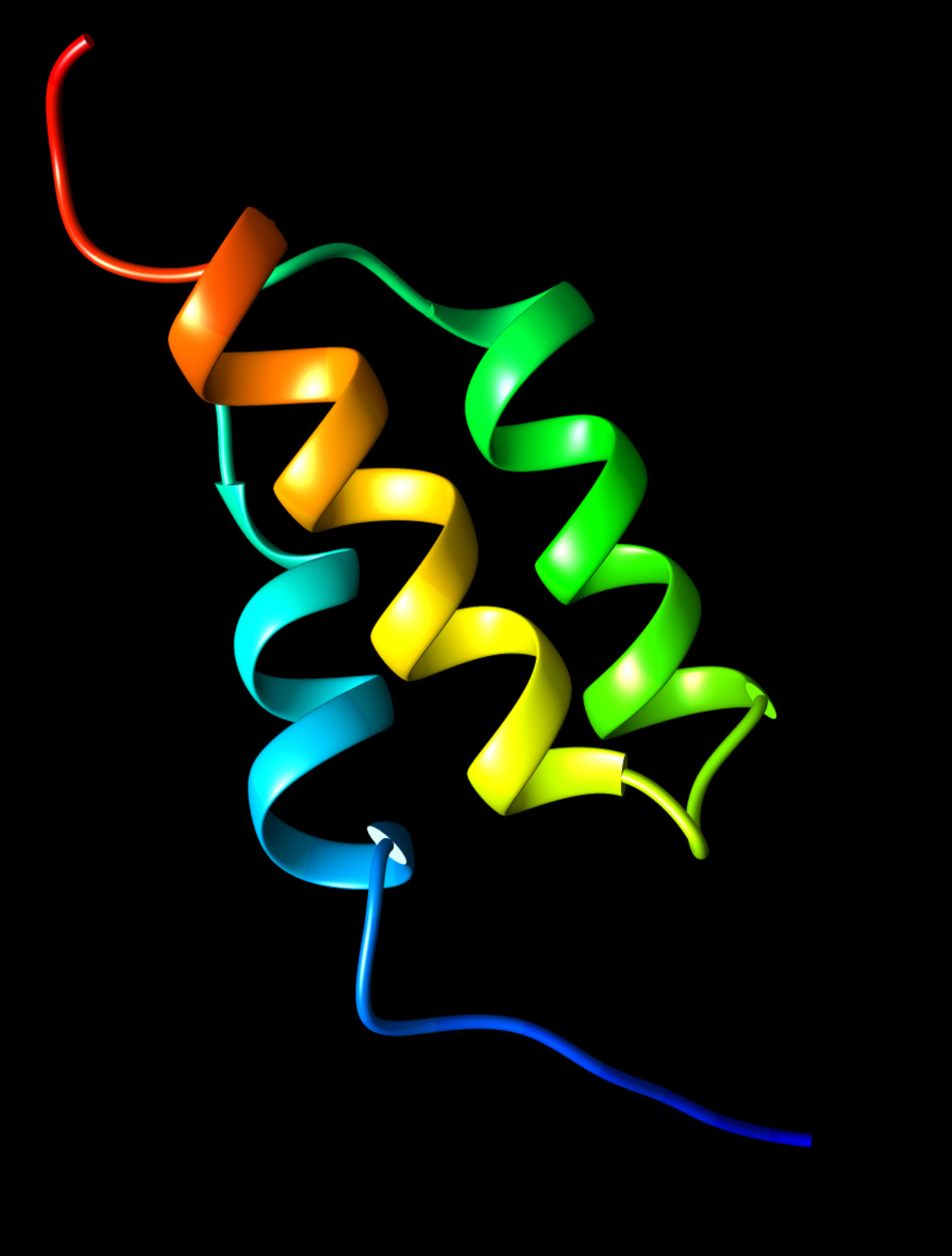}
	\end{center}
\vspace{0.25in}
\hspace*{3in}
{\Large
	\begin{minipage}[t]{3in}
	\baselineskip = .5\baselineskip
	Figure \ref{fig:1bdd} \\
	Johnston, Zhang, Liwo, Crivelli, and Taufer \\
	J.\ Comput.\ Chem.
	\end{minipage}
}

\vspace*{0.1in}   
	\begin{center}
		\includegraphics[width=0.8\columnwidth,keepaspectratio=true]{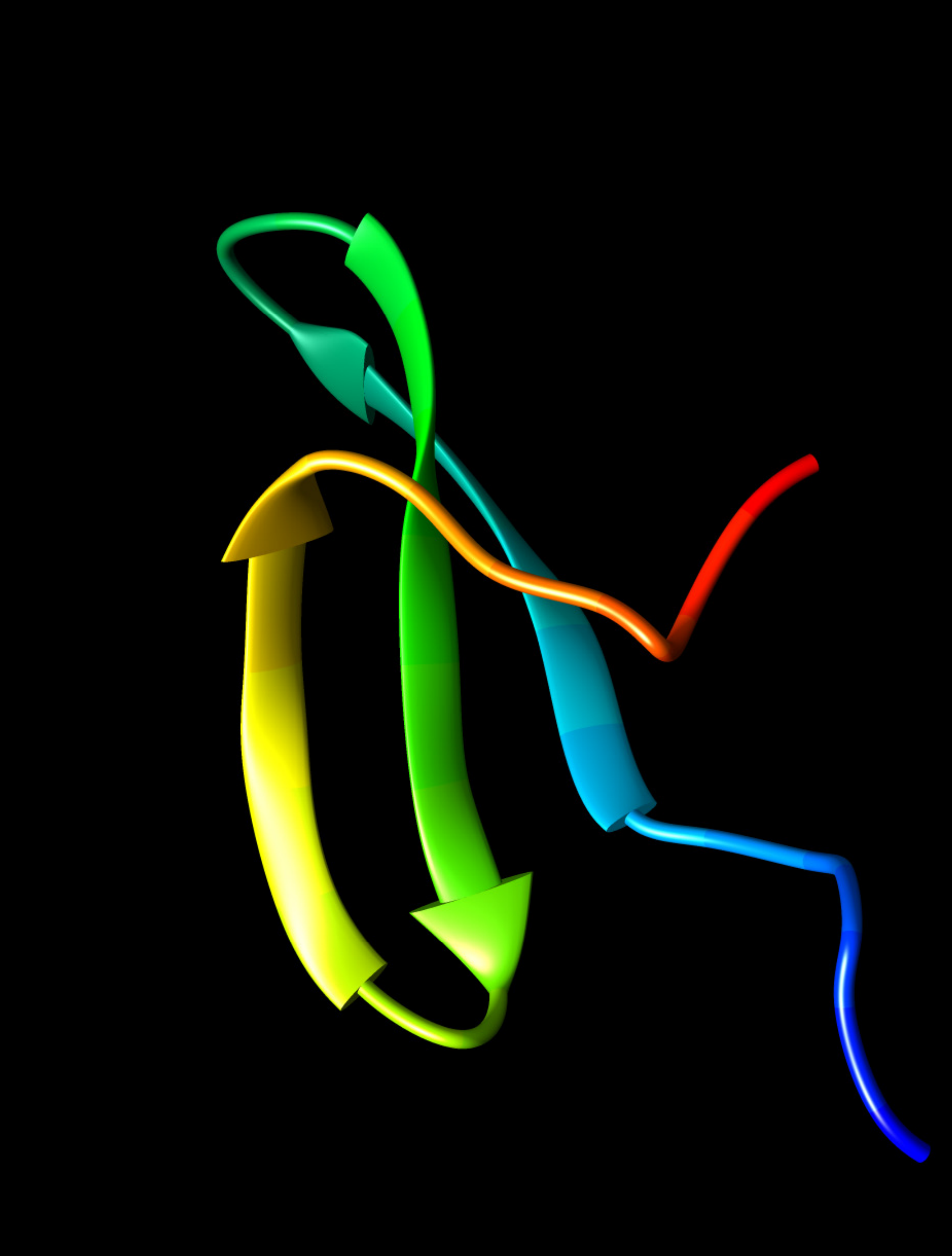}
	\end{center}
\vspace{0.25in}
\hspace*{3in}
{\Large
	\begin{minipage}[t]{3in}
	\baselineskip = .5\baselineskip
	Figure \ref{fig:1e0l} \\
	Johnston, Zhang, Liwo, Crivelli, and Taufer \\
	J.\ Comput.\ Chem.
	\end{minipage}
}

\clearpage

	\begin{minipage}[t]{\textwidth}
		\centering
		\includegraphics[width=0.85\columnwidth,keepaspectratio=true]{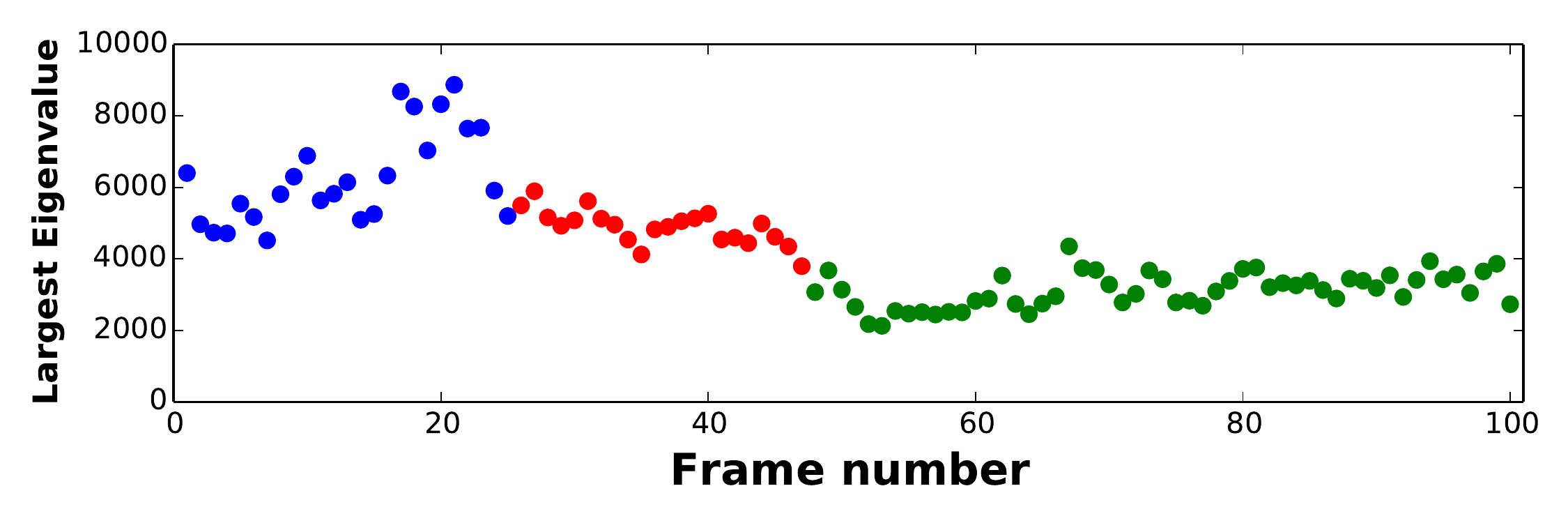}\hspace{.75 in}	

		(a)
	\end{minipage}
	\begin{minipage}[t]{\textwidth}
		\centering
		\includegraphics[width=0.8\columnwidth,keepaspectratio=true]{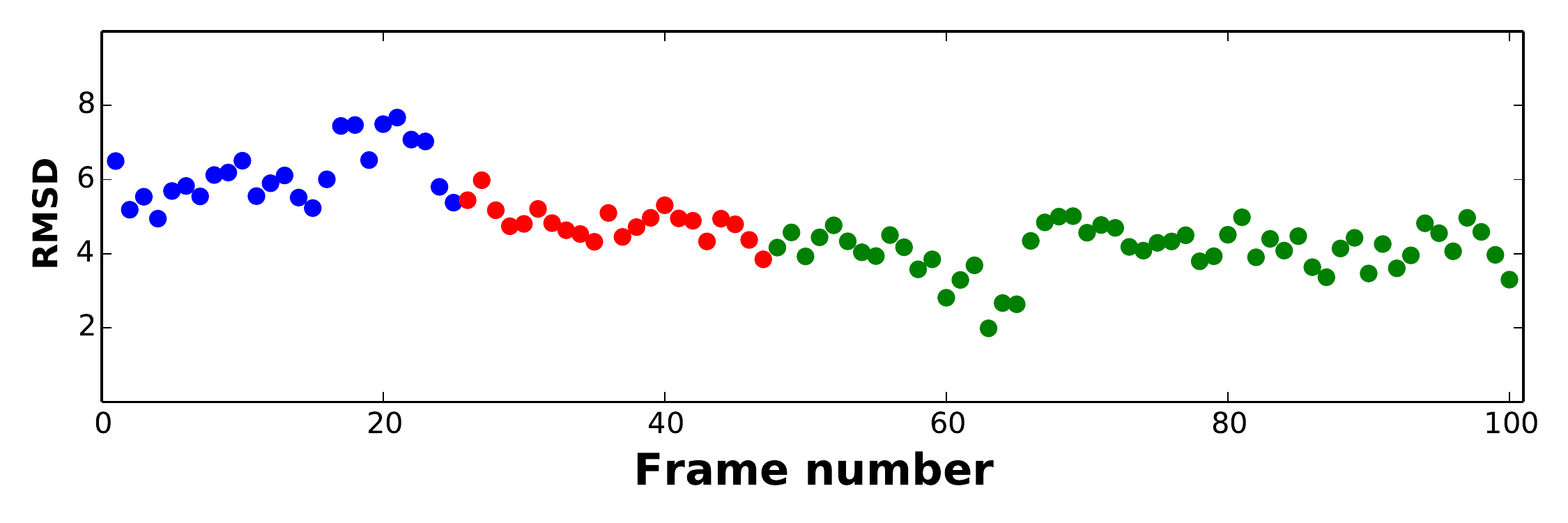}
	
		(b)
	\end{minipage}
	\begin{minipage}[t]{\textwidth}
		\centering
		\includegraphics[width=0.8\columnwidth,keepaspectratio=true]{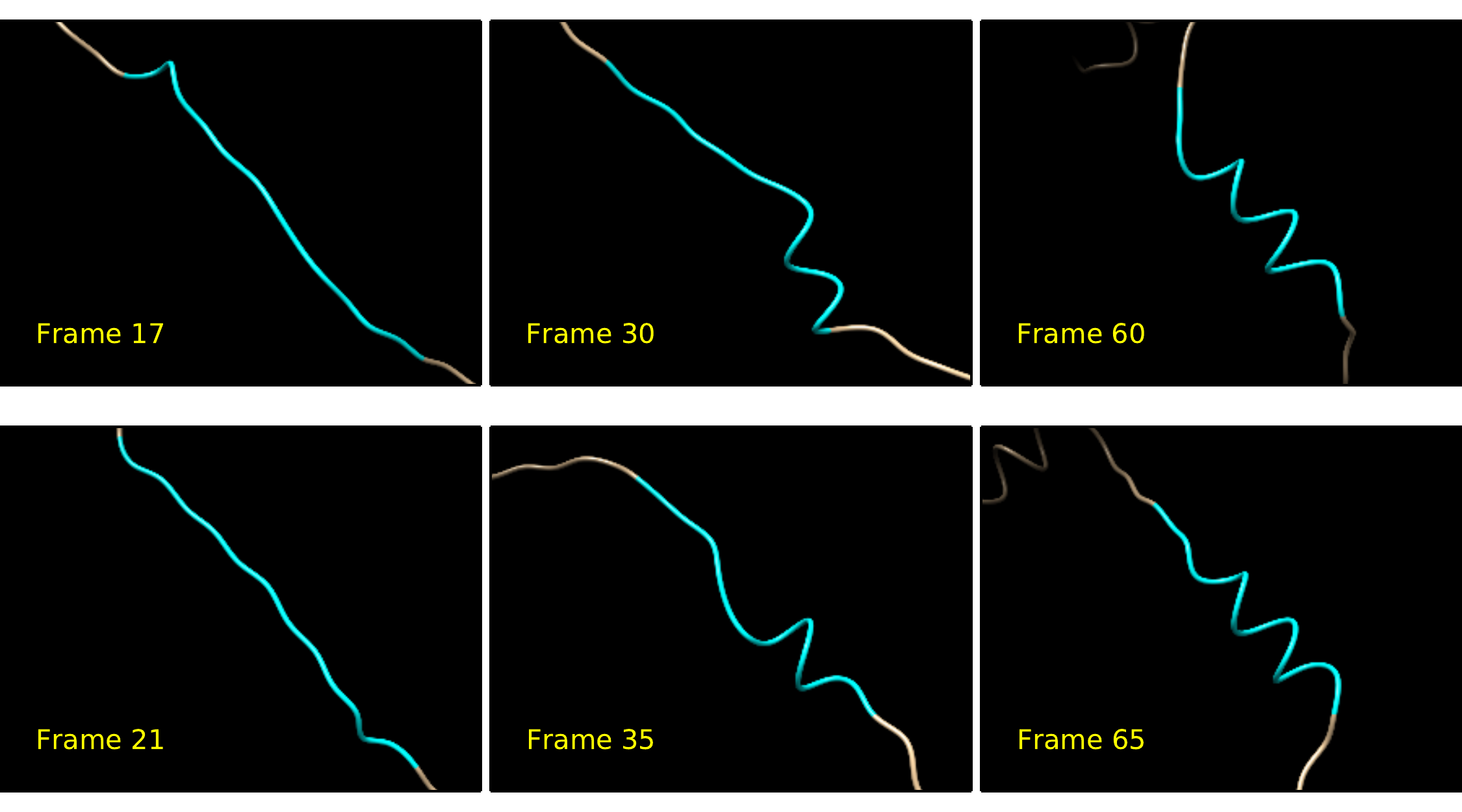}
	
		(c)
	\end{minipage}
\vspace{0.25in}
\hspace*{3in}
{\Large
	\begin{minipage}[t]{3in}
	\baselineskip = .5\baselineskip
	Figure \ref{helix2_formation} \\
	Johnston, Zhang, Liwo, Crivelli, and Taufer \\
	J.\ Comput.\ Chem.
	\end{minipage}
}

	\begin{minipage}[t]{\textwidth}
		\centering
		\includegraphics[width=0.85\columnwidth,keepaspectratio=true]{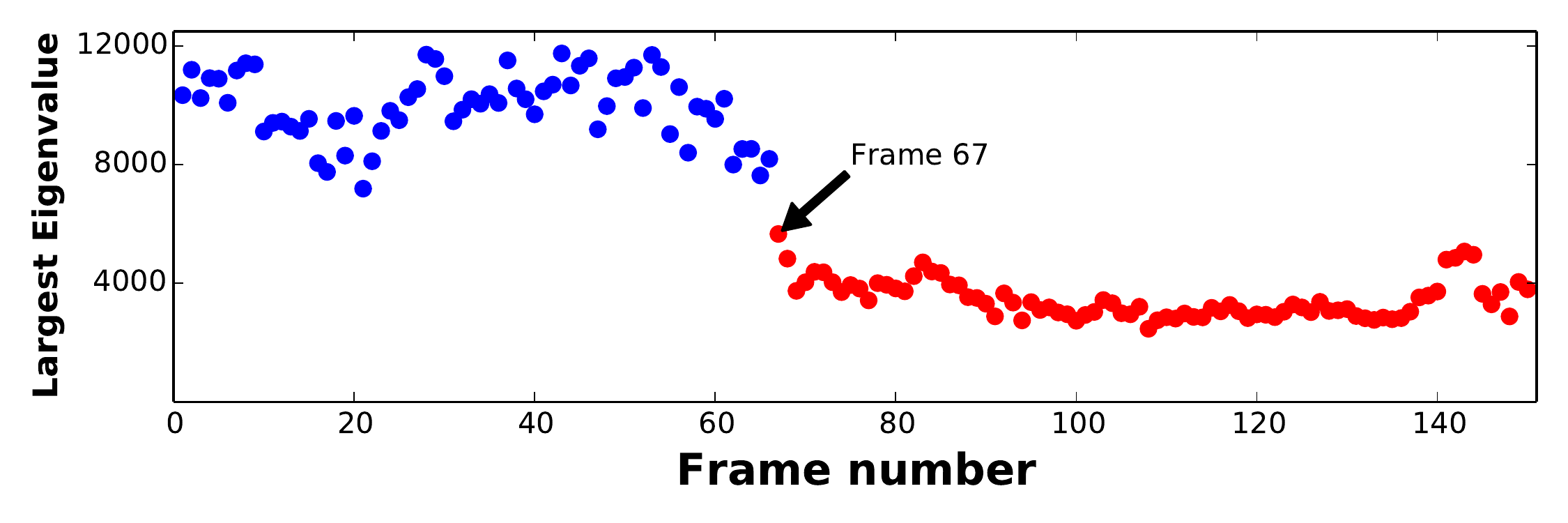}\hspace{.75 in}
	
		(a)
	\end{minipage}
	\begin{minipage}[t]{\textwidth}
		\centering
		\includegraphics[width=0.8\columnwidth,keepaspectratio=true]{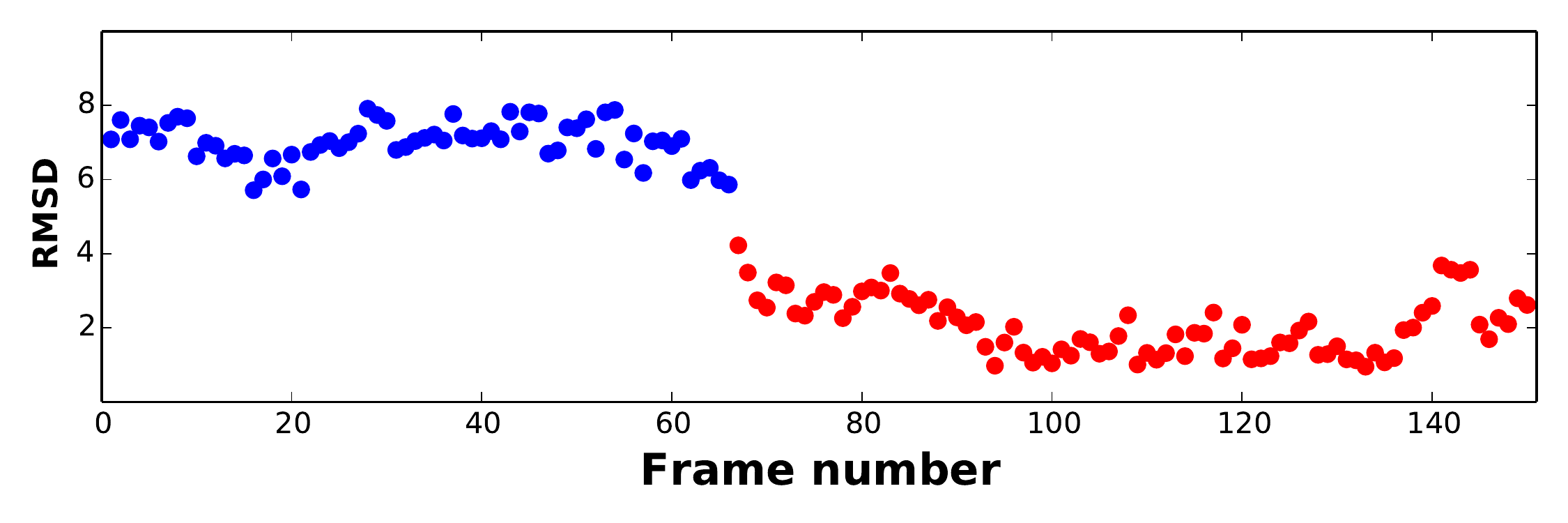}

		(b)
	\end{minipage}
	\begin{minipage}[t]{\textwidth}
		\centering
		\includegraphics[width=0.8\columnwidth,keepaspectratio=true]{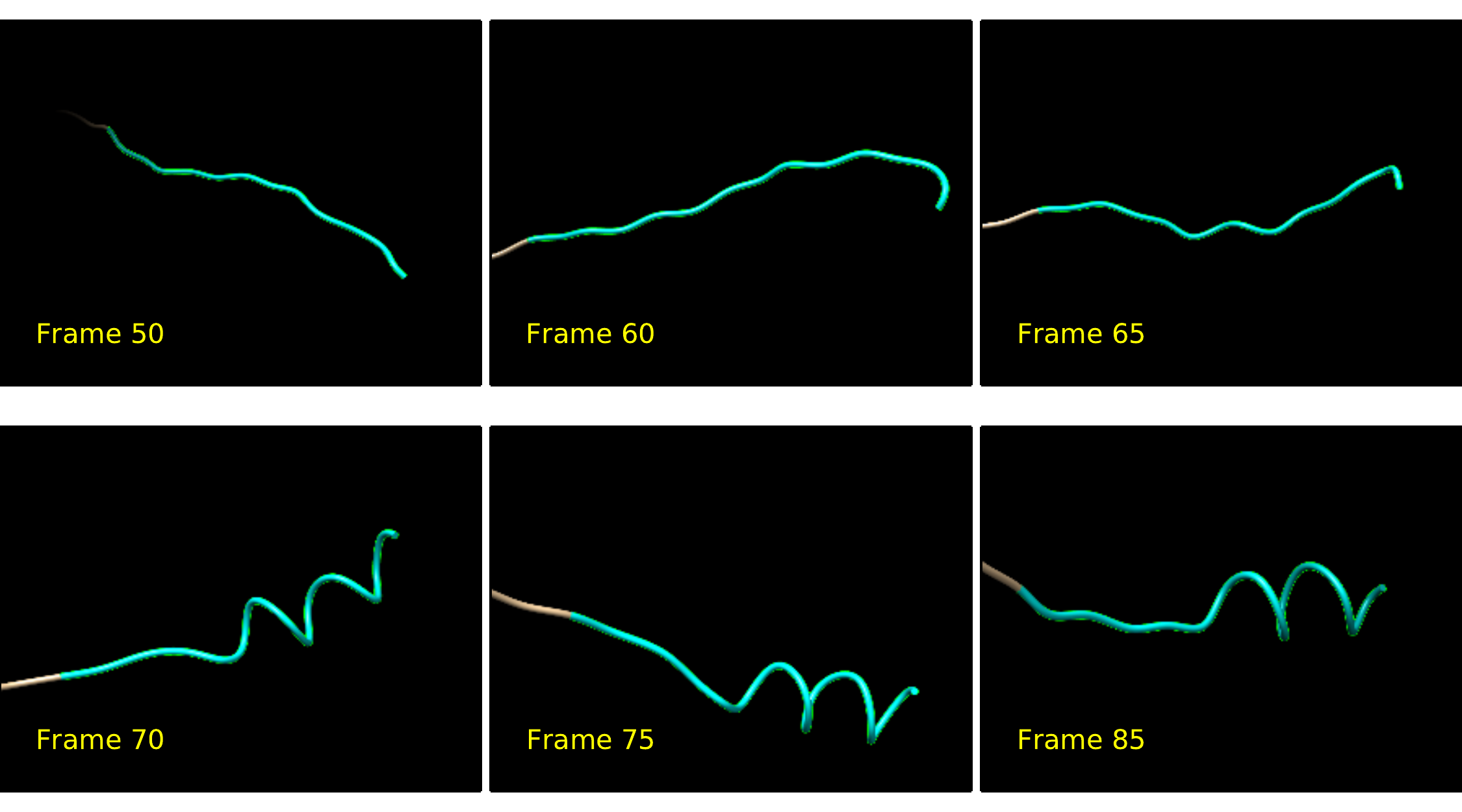}		

		(c)
	\end{minipage}
\vspace{0.25in}
\hspace*{3in}
{\Large
	\begin{minipage}[t]{3in}
	\baselineskip = .5\baselineskip
	Figure \ref{helix3_formation} \\
	Johnston, Zhang, Liwo, Crivelli, and Taufer \\
	J.\ Comput.\ Chem.
	\end{minipage}
}

\vspace*{0.1in}   
	\begin{minipage}[t]{\textwidth}
	\centering
		\includegraphics[width=0.9\columnwidth,keepaspectratio=true]{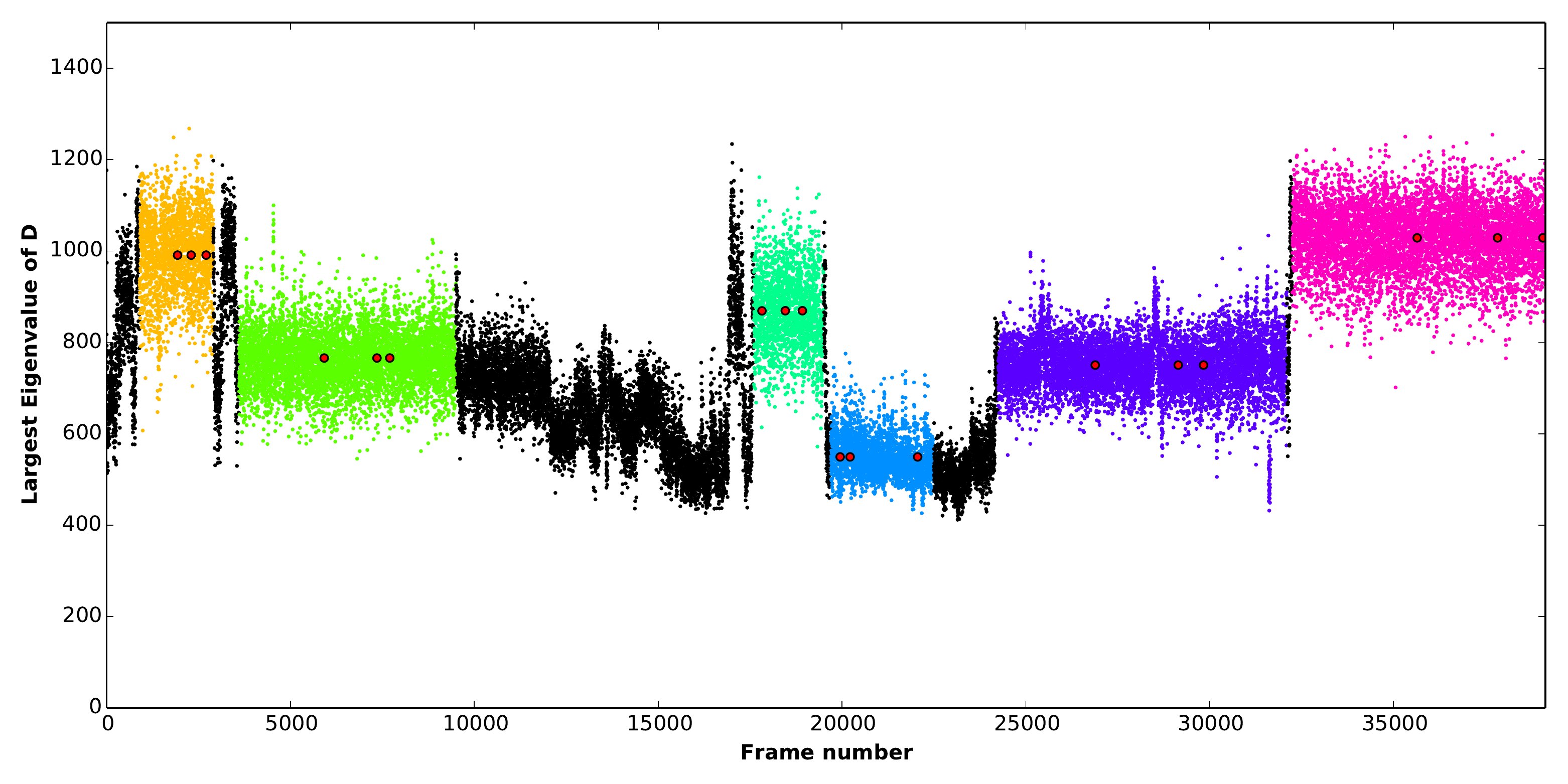}

		(a)
	\end{minipage}
	\begin{minipage}[t]{\textwidth}
	\centering
		\includegraphics[width=0.9\columnwidth,keepaspectratio=true]{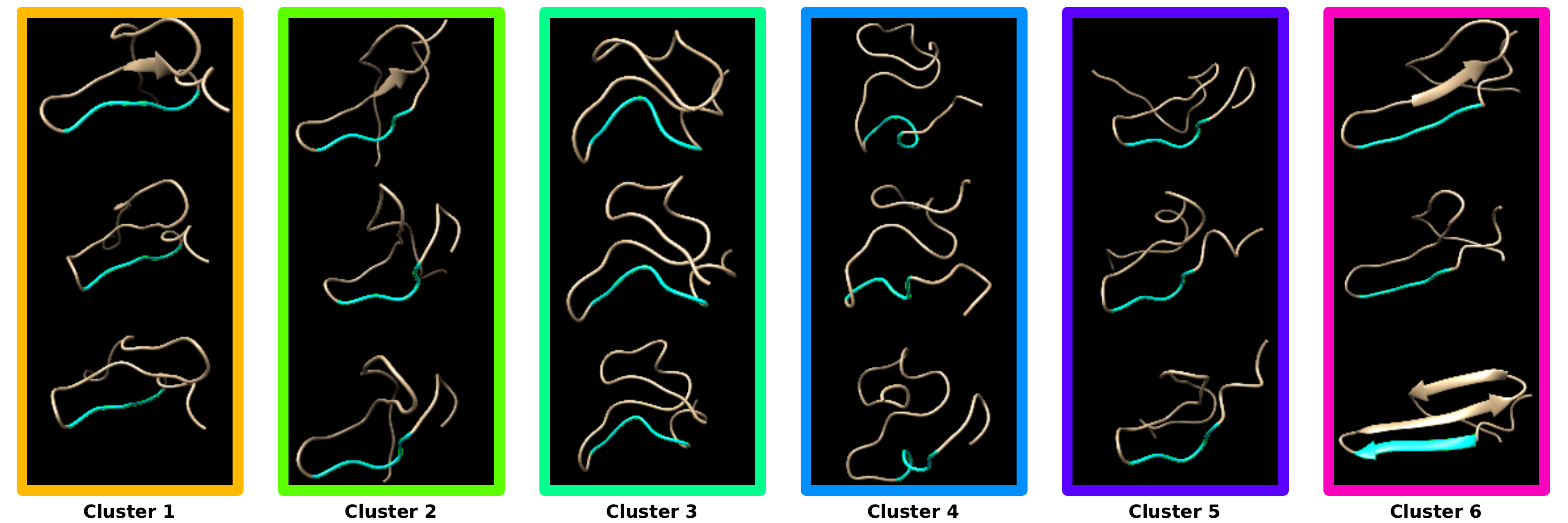}

		(b)
	\end{minipage}
\vspace{0.25in}
\hspace*{3in}
{\Large
	\begin{minipage}[t]{3in}
	\baselineskip = .5\baselineskip
	Figure \ref{sheet1_validation} \\
	Johnston, Zhang, Liwo, Crivelli, and Taufer \\
	J.\ Comput.\ Chem.
	\end{minipage}
}

\vspace*{0.1in}
	\begin{minipage}[t]{\textwidth}
	\centering
		\includegraphics[width=0.9\columnwidth,keepaspectratio=true]{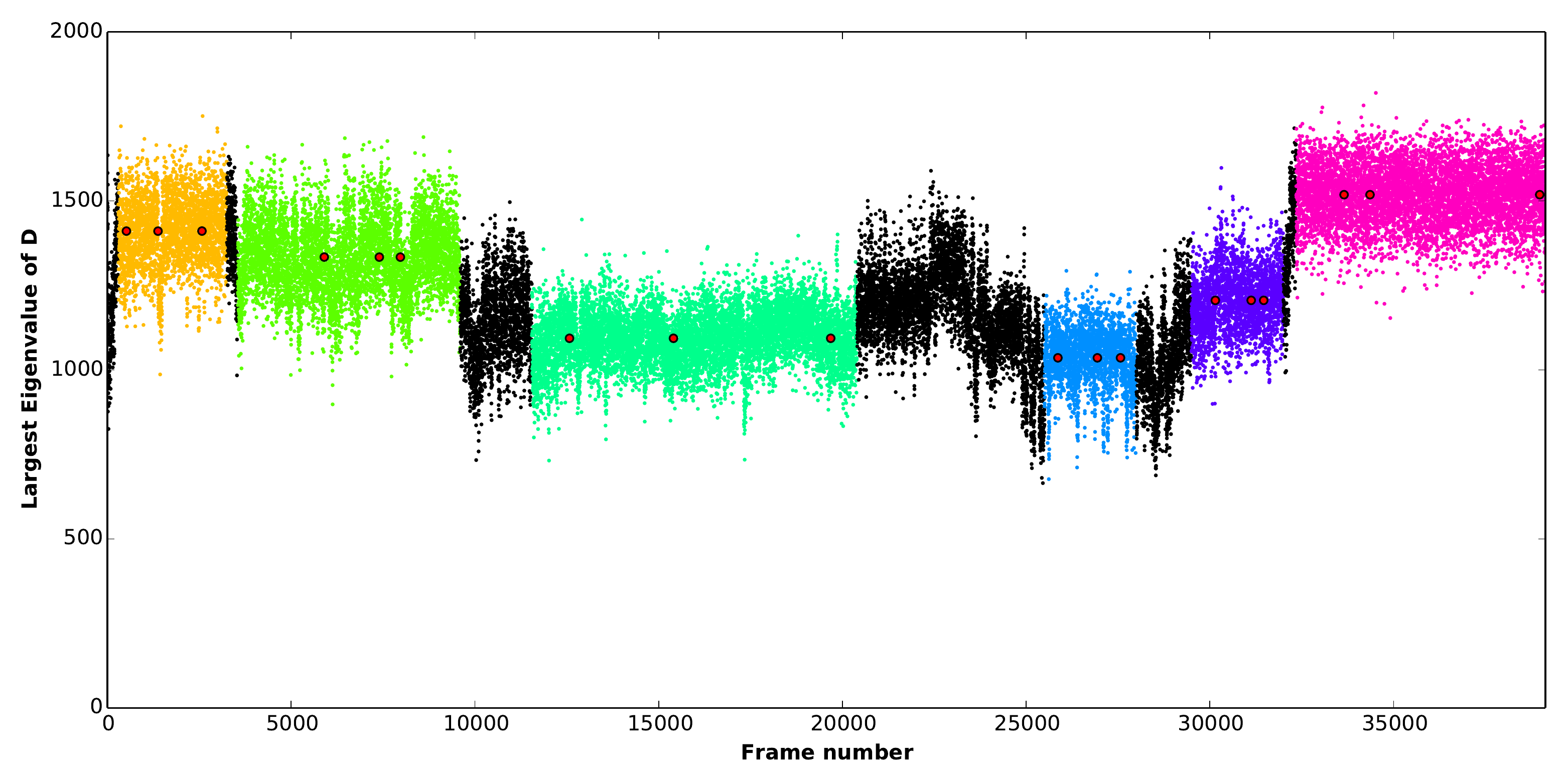}

		(a)
	\end{minipage}
	\begin{minipage}[t]{\textwidth}
	\centering
		\includegraphics[width=0.9\columnwidth,keepaspectratio=true]{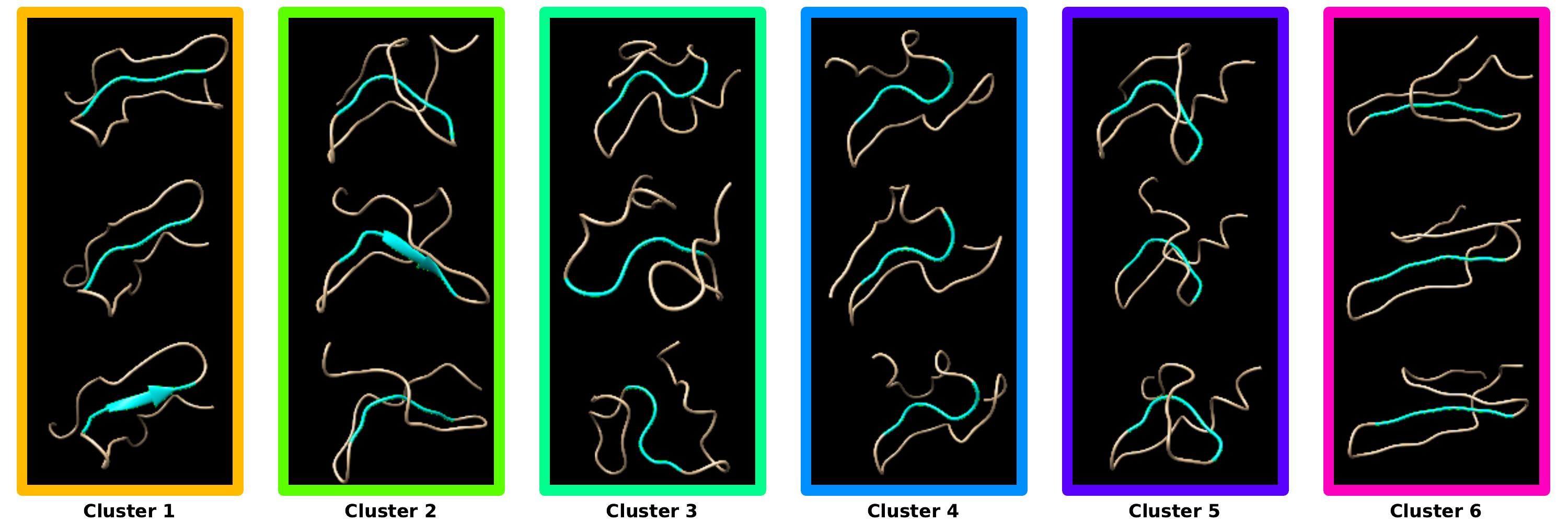}

		(b)
	\end{minipage}
\vspace{0.25in}
\hspace*{3in}
{\Large
	\begin{minipage}[t]{3in}
	\baselineskip = .5\baselineskip
	Figure \ref{sheet2_validation} \\
	Johnston, Zhang, Liwo, Crivelli, and Taufer \\
	J.\ Comput.\ Chem.
	\end{minipage}
}

\vspace*{0.1in}   
	\begin{minipage}[t]{\textwidth}
	\centering
		\includegraphics[width=0.9\columnwidth,keepaspectratio=true]{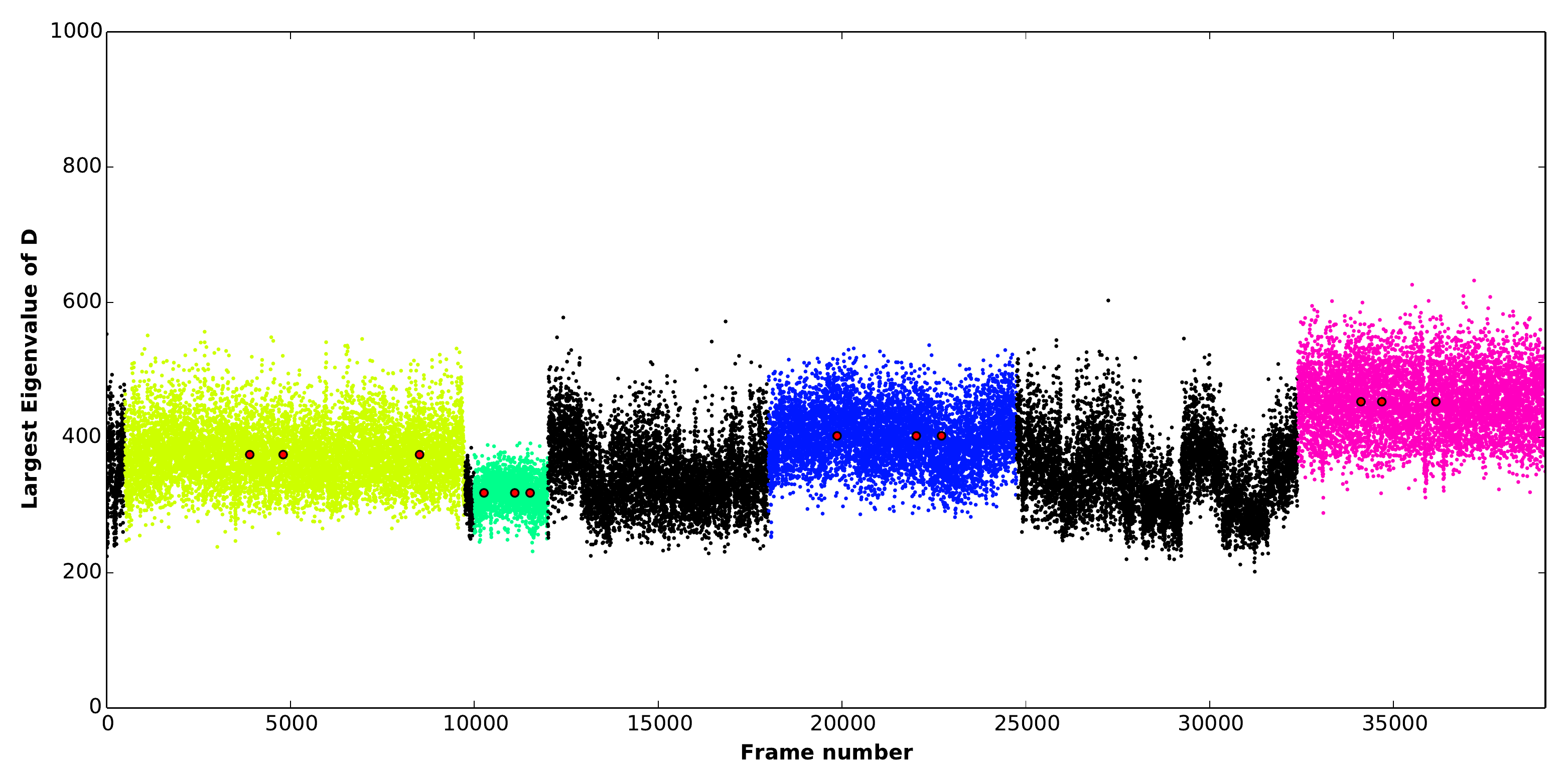}

		(a)
	\end{minipage}
	\begin{minipage}[t]{\textwidth}
	\centering
		\includegraphics[width=0.9\columnwidth,keepaspectratio=true]{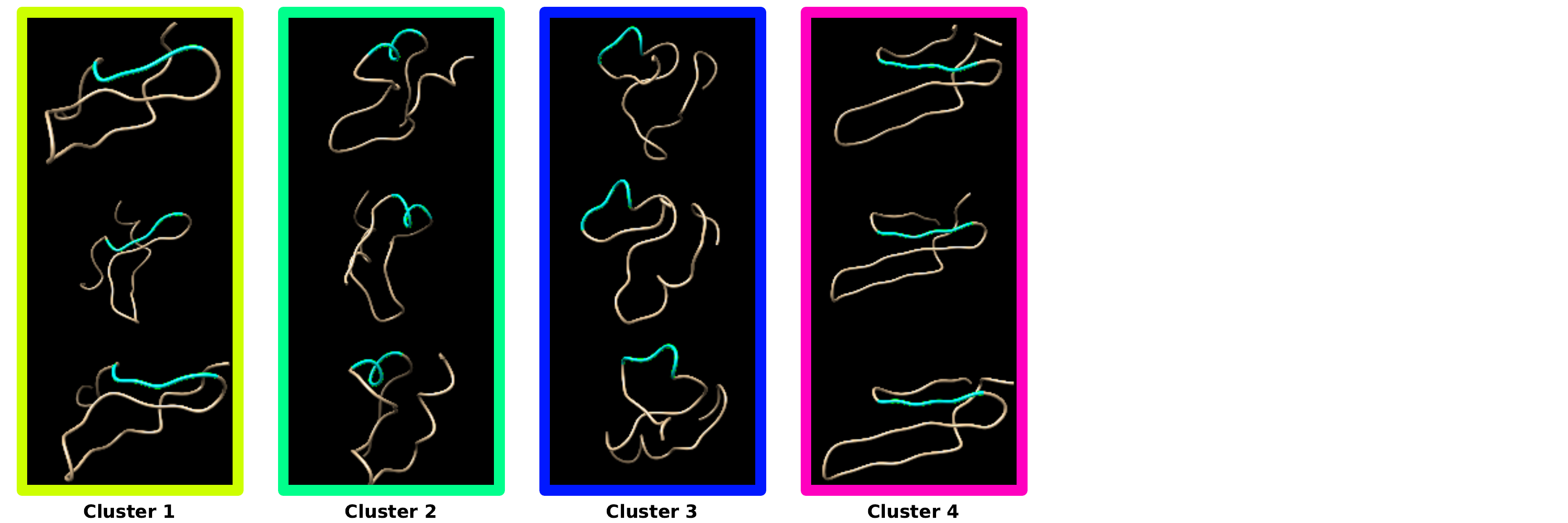}

		(b)
	\end{minipage}
\vspace{0.25in}
\hspace*{3in}
{\Large
	\begin{minipage}[t]{3in}
	\baselineskip = .5\baselineskip
	Figure \ref{sheet3_validation} \\
	Johnston, Zhang, Liwo, Crivelli, and Taufer \\
	J.\ Comput.\ Chem.
	\end{minipage}
}

\vspace*{0.1in}   
	\begin{center}
		\includegraphics[width=0.9\columnwidth,keepaspectratio=true]{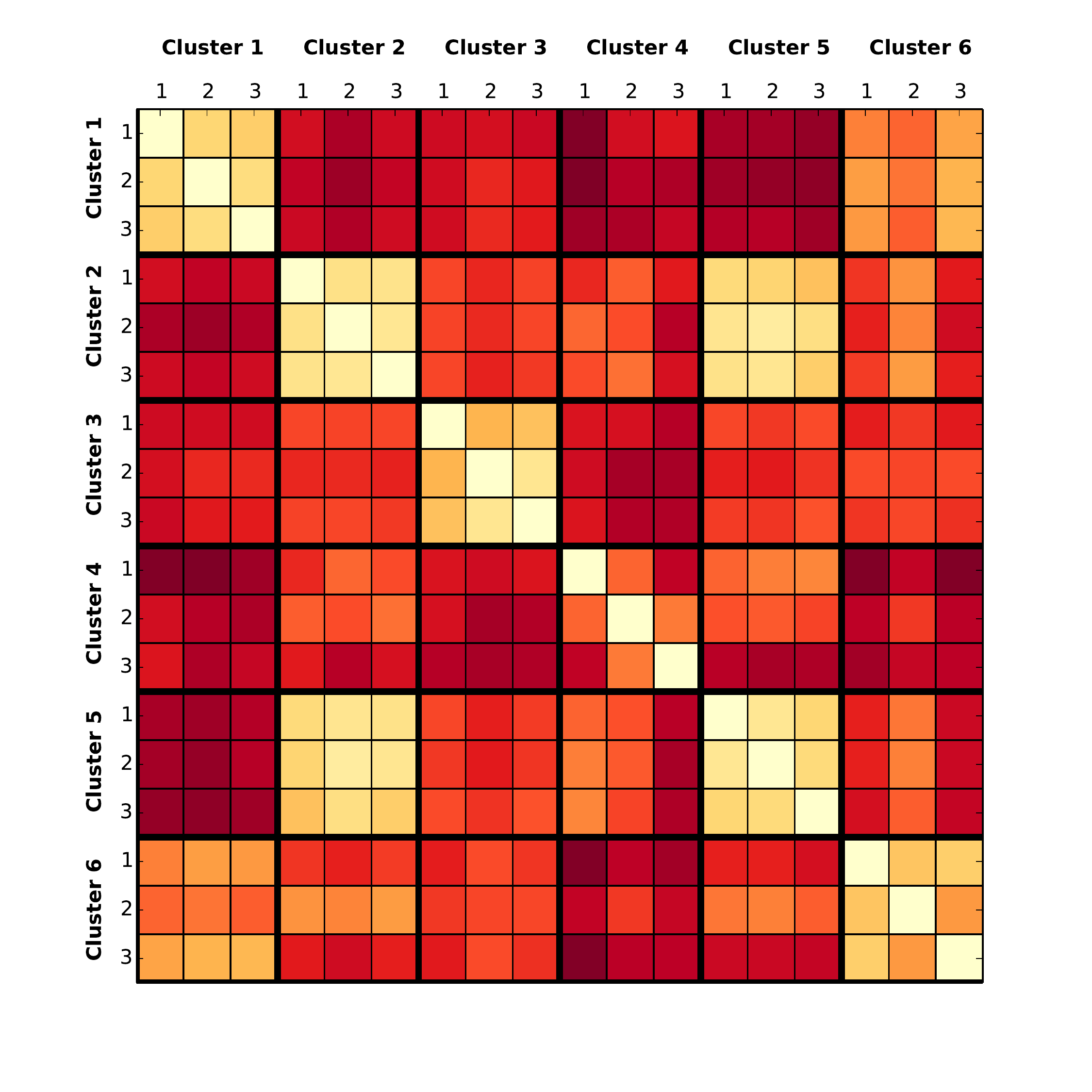}
	\end{center}
\vspace{0.25in}
\hspace*{3in}
{\Large
	\begin{minipage}[t]{3in}
	\baselineskip = .5\baselineskip
	Figure \ref{sheet1_heatmap} \\
	Johnston, Zhang, Liwo, Crivelli, and Taufer \\
	J.\ Comput.\ Chem.
	\end{minipage}
}

\vspace*{0.1in}   
	\begin{center}
		\includegraphics[width=0.9\columnwidth,keepaspectratio=true]{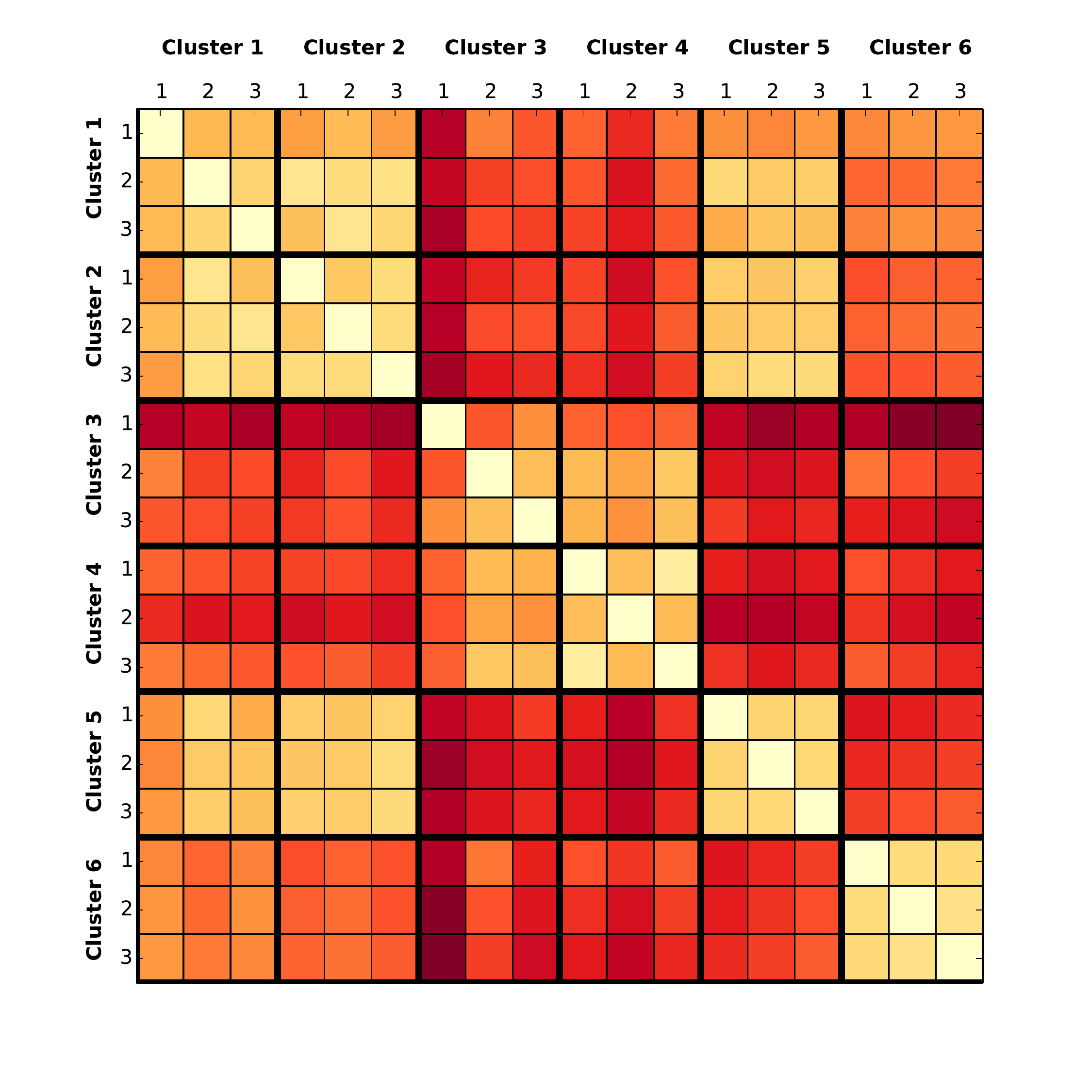}
	\end{center}
\vspace{0.25in}
\hspace*{3in}
{\Large
	\begin{minipage}[t]{3in}
	\baselineskip = .5\baselineskip
	Figure \ref{sheet2_heatmap} \\
	Johnston, Zhang, Liwo, Crivelli, and Taufer \\
	J.\ Comput.\ Chem.
	\end{minipage}
}

\vspace*{0.1in}   
	\begin{center}
		\includegraphics[width=0.9\columnwidth,keepaspectratio=true]{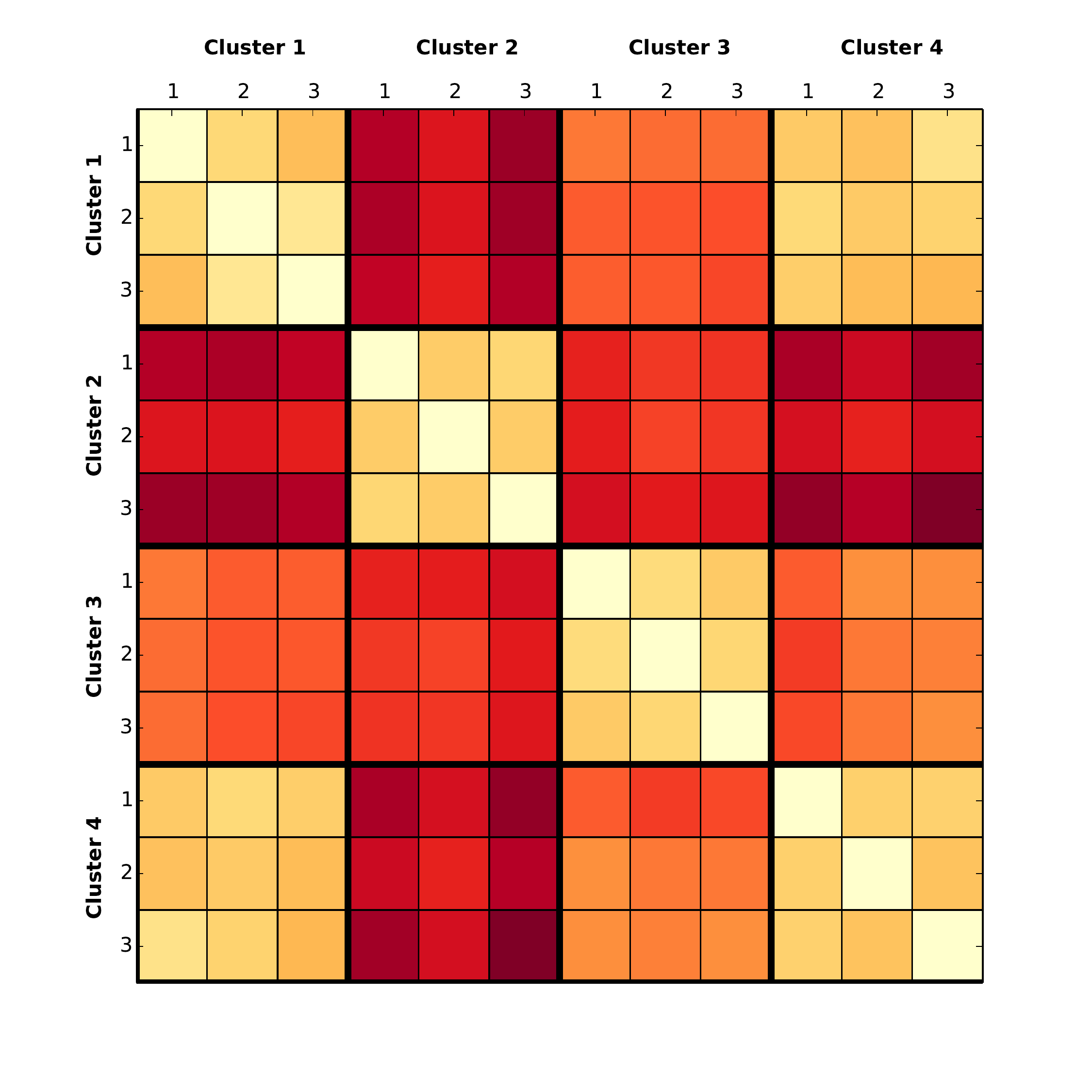}
	\end{center}
\vspace{0.25in}
\hspace*{3in}
{\Large
	\begin{minipage}[t]{3in}
	\baselineskip = .5\baselineskip
	Figure \ref{sheet3_heatmap} \\
	Johnston, Zhang, Liwo, Crivelli, and Taufer \\
	J.\ Comput.\ Chem.
	\end{minipage}
}

\clearpage

\vspace*{0.1in}
	\begin{minipage}[t]{\textwidth}
	\centering
		\includegraphics[width=0.75\columnwidth,keepaspectratio=true]{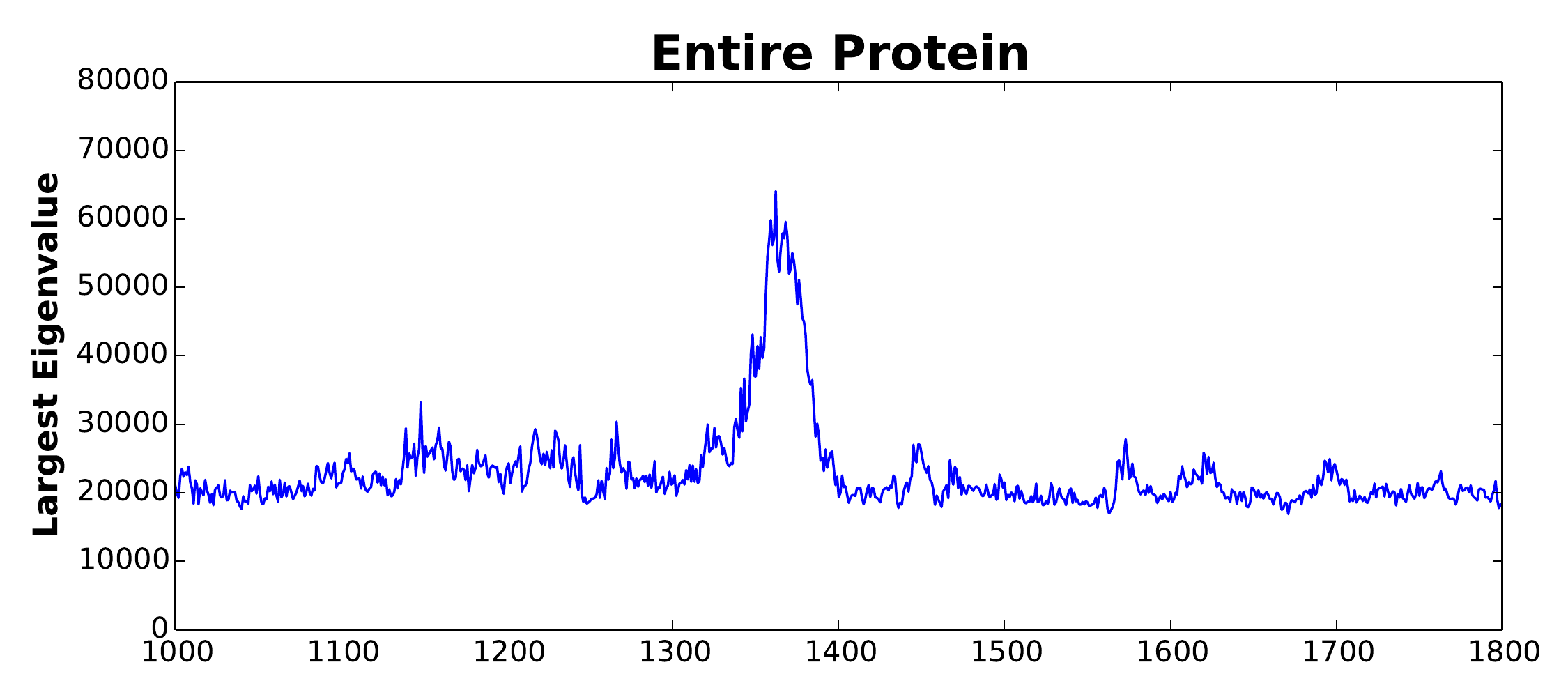}\hspace{.5 in}

		(a)
	\end{minipage}
	\begin{minipage}[t]{\textwidth}
	\centering
		\includegraphics[width=0.7\columnwidth,keepaspectratio=true]{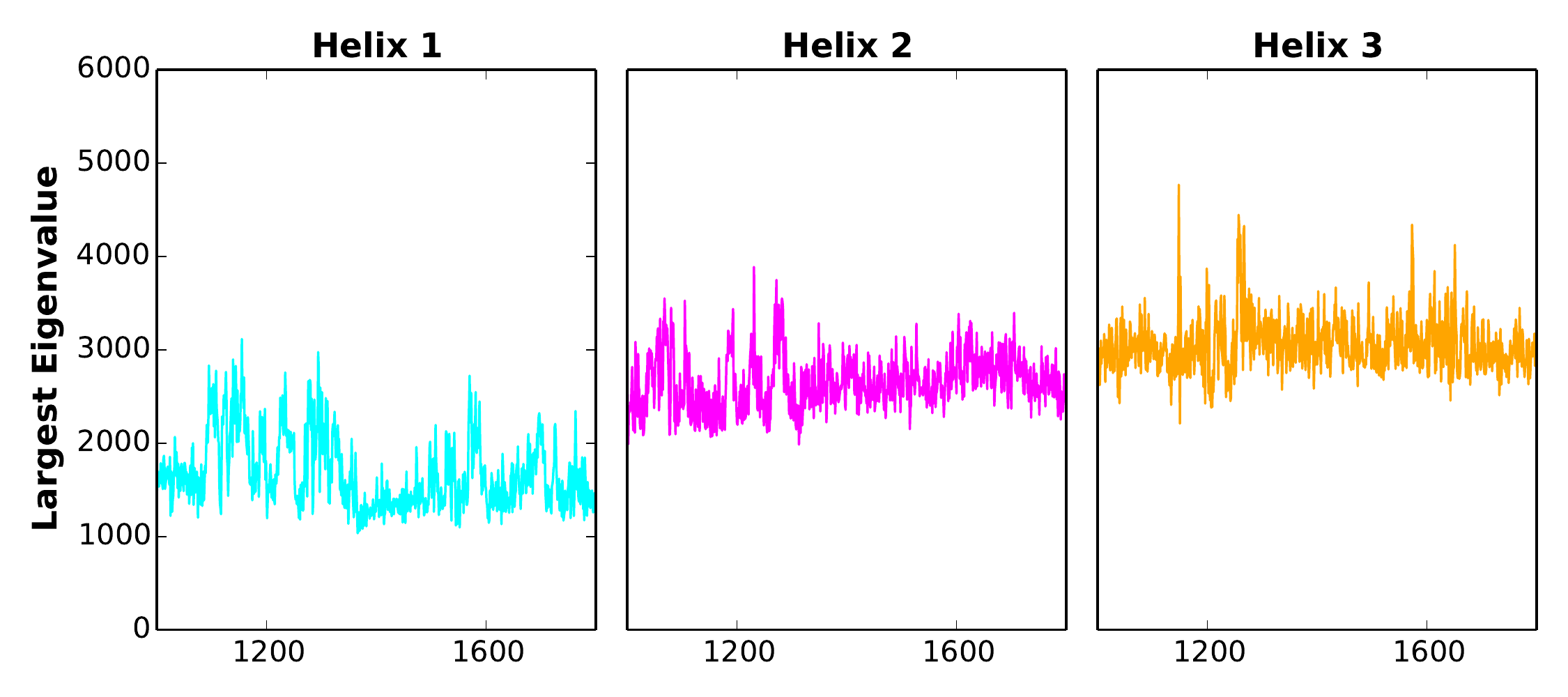}

		(b)
	\end{minipage}
	\begin{minipage}[t]{\textwidth}
	\centering
		\includegraphics[width=0.7\columnwidth,keepaspectratio=true]{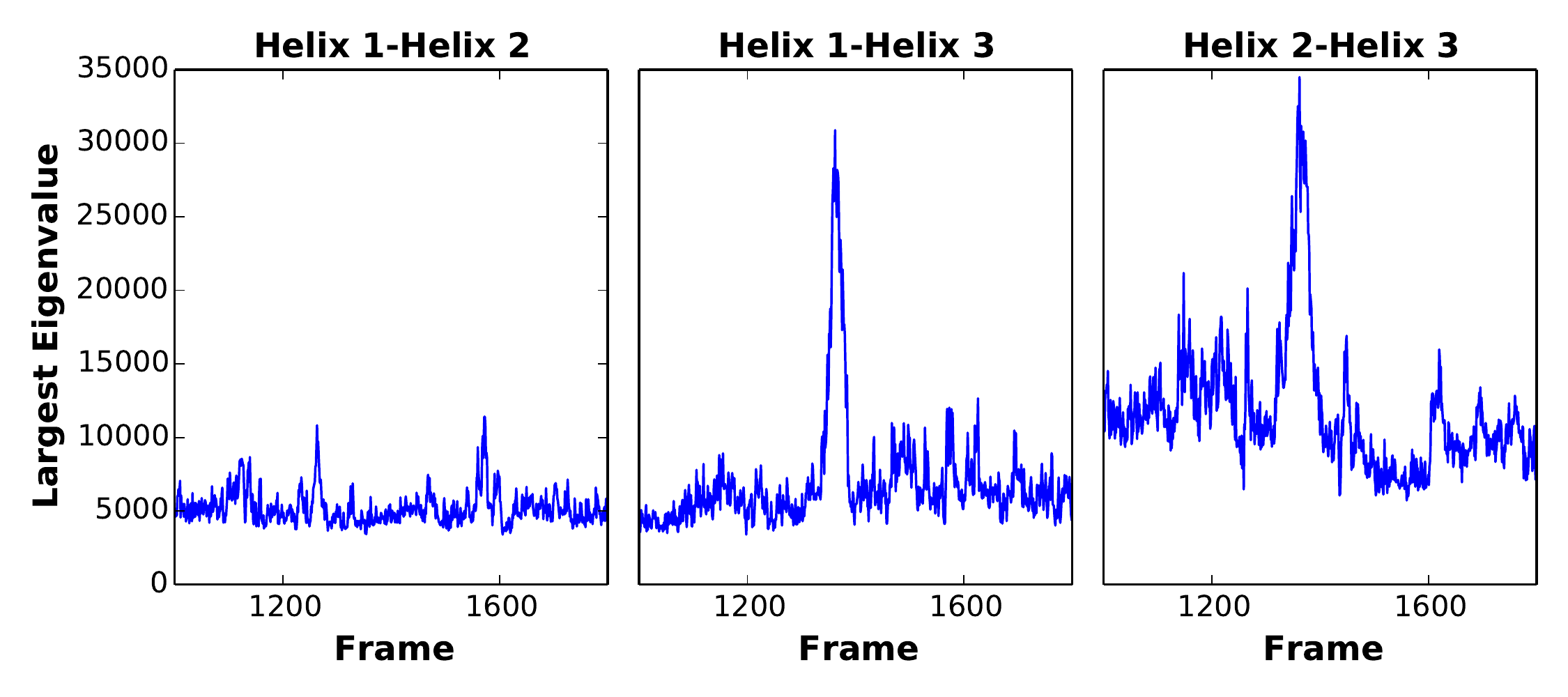}

		(c)
	\end{minipage}
\vspace{0.25in}
\hspace*{3in}
{\Large
	\begin{minipage}[t]{3in}
	\baselineskip = .5\baselineskip
	Figure \ref{ev_helix3_moving} \\
	Johnston, Zhang, Liwo, Crivelli, and Taufer \\
	J.\ Comput.\ Chem.
	\end{minipage}
}

\clearpage

\vspace*{0.1in}   
	\begin{center}
		\includegraphics[width=0.8\columnwidth,keepaspectratio=true]{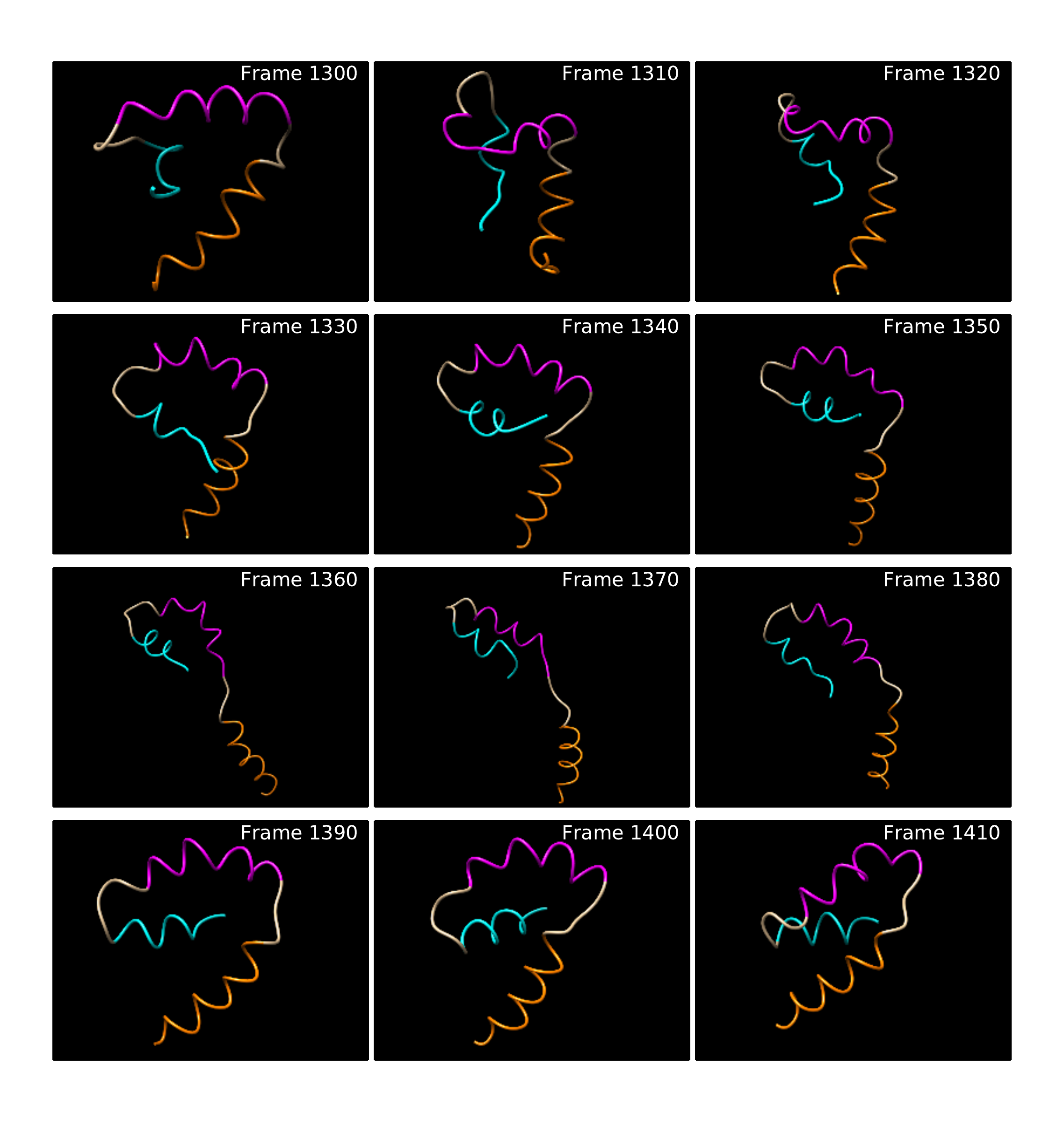}
	\end{center}
\vspace{0.25in}
\hspace*{3in}
{\Large
	\begin{minipage}[t]{3in}
	\baselineskip = .5\baselineskip
	Figure \ref{helix3_moving} \\
	Johnston, Zhang, Liwo, Crivelli, and Taufer \\
	J.\ Comput.\ Chem.
	\end{minipage}
}

\end{document}